\documentclass[journal,twoside,web]{ieeecolor}

\usepackage{generic}
\usepackage{graphicx,wrapfig,hyperref}
\usepackage[utf8]{inputenc}
\usepackage{mathtools}
\usepackage{inputenc}
\usepackage[english]{babel}
\usepackage{amsmath}
\usepackage{amssymb}
\usepackage{amsfonts}
\usepackage{cleveref}
\usepackage{comment}
\usepackage{bm}
\usepackage{psfrag}
\usepackage{esdiff}
\usepackage{float}
\usepackage{subcaption}
\usepackage{wrapfig}
\usepackage{bm}
\usepackage{algorithmicx}
\usepackage{algpseudocode}
\usepackage{lipsum}
\usepackage[linesnumbered,ruled,vlined]{algorithm2e}
\usepackage{array}
\usepackage{blindtext}
\usepackage{epstopdf}
\usepackage{ragged2e}
\usepackage{mathrsfs}
\usepackage{tabu}
\graphicspath{ {RAL_interception/figures_interception/} }

\usepackage{url}
\usepackage{nomencl}
\usepackage{cite}
\usepackage{lmodern}
\usepackage{enumerate}
\usepackage{tikz}
\usepackage[switch]{lineno}
\usepackage{psfrag}
\usepackage{pstool}
\usepackage{slashbox}
\usepackage{multirow}
\usepackage{etoolbox}

\DeclareMathOperator*{\argmin}{arg\,min}

\newcommand{\be}{\begin{equation}}
	\newcommand{\ee}{\end{equation}}
\newcommand{\bse}{\begin{subequations}}
	\newcommand{\ese}{\end{subequations}}
\newcommand{\bewn}{\begin{equation*}}
	\newcommand{\eewn}{\end{equation*}}
\newcommand{\bbmat}{\begin{bmatrix}} 
	\newcommand{\ebmat}{\end{bmatrix}}
\newcommand{\bd}{\begin{displaymath}}
	\newcommand{\ed}{\end{displaymath}}
\newcommand{\bea}{\begin{eqnarray}}
	\newcommand{\eea}{\end{eqnarray}}
\newcommand{\ba}{\begin{array}}
	\newcommand{\ea}{\end{array}}
\newcommand{\baa}{\begin{array}{ll}}
	\newcommand{\eaa}{\end{array}}
\newcommand{\ds}{\displaystyle}
\newcommand{\bc}{\begin{center}}
	\newcommand{\ec}{\end{center}}
\newcommand{\ben}{\begin{enumerate}}
	\newcommand{\een}{\end{enumerate}}
\newcommand{\bi}{\begin{itemize}}
	\newcommand{\ei}{\end{itemize}}
\newcommand{\bt}{\begin{tabular}}
	\newcommand{\et}{\end{tabular}}
\newcommand{\bte}{\begin{table}}
	\newcommand{\ete}{\end{table}}
\newcommand{\bal}{\begin{align}}
	\newcommand{\eal}{\end{align}}
	
\renewcommand{\baselinestretch}{1}

\SetKwFunction{trianglesIntersect}{triInt}
\SetKwFunction{collisionTime}{collisionTime}

\newcommand{\norm}[1]{\left\lVert#1\right\rVert}   

\makeatletter
\renewcommand\paragraph{\@startsection{paragraph}{4}{\z@}%
	{-2.5ex\@plus -1ex \@minus -.25ex}%
	{1.25ex \@plus .25ex}%
	{\normalfont\normalsize\bfseries}}
\makeatother



\makeatletter
\patchcmd{\@algocf@start}
  {-1.5em}
  {0pt}
  {}{}
\makeatother

\newtheorem{remark}{Remark}

\newtheorem{condition}{Condition}
\newtheorem{theorem}{Theorem}

\newtheorem{lemma}[theorem]{\textbf{Lemma}}
\newtheorem{assumption}{\textbf{Assumption}}
\newtheorem{problem}{\textbf{Problem}}
\newtheorem{definition}{\textbf{Definition}}

\newcommand{\bR}{\mathbb{R}}

\newcommand{\calA}{\mathcal{A}}
\newcommand{\calB}{\mathcal{B}}

\newcommand{\calD}{\mathcal{D}}

\newcommand{\calF}{\mathcal{F}}

\newcommand{\calP}{\mathcal{P}}

\newcommand{\calR}{\mathcal{R}}

\newcommand{\calW}{\mathcal{W}}
\newcommand{\calX}{\mathcal{X}}

\newcommand{\abs}[1]{\left |#1\right |}

\newcommand{\vishnu}[1]{\textcolor{black}{#1}}

\newcommand{\xinyi}[1]{\textcolor{black}{#1}}

\begin{document}
	\title{\LARGE \bf
		IDCAIS: Inter-Defender Collision-Aware Interception Strategy against Multiple Attackers
	}

	\author{Vishnu S. Chipade, Xinyi Wang and Dimitra Panagou
		\thanks{Vishnu Chipade is now an independent researcher; this work was mostly conducted when he was a PhD student at the Department of Aerospace Engineering, University of Michigan. Xinyi Wang and Dimitra Panagou are with the Department of Robotics and the Department of Aerospace Engineering, University of Michigan, Ann Arbor, MI, USA; 
			{\tt\small (vishnuc, xinywa,
				dpanagou)@umich.edu}}
		\thanks{This work has been funded by the Center for Unmanned Aircraft Systems (C-UAS), a National Science Foundation Industry/University Cooperative Research Center (I/UCRC) under NSF Award No. 1738714 along with significant contributions from C-UAS industry members.}
	}	
	
	\maketitle
	\thispagestyle{empty}
	\pagestyle{empty}

	\begin{abstract}
	In the prior literature on multi-agent area defense games, the assignments of the defenders to the attackers are done based on a cost metric associated only with the interception of the attackers. In contrast to that, this paper presents an Inter-Defender Collision-Aware Interception Strategy (IDCAIS) for defenders to intercept attackers in order to defend a protected area, such that the defender-to-attacker assignment protocol not only takes into account an interception-related cost but also takes into account any possible future collisions among the defenders on their optimal interception trajectories. In particular, in this paper, the defenders are assigned to intercept attackers using a mixed-integer quadratic program (MIQP) that: 1) minimizes the sum of times taken by defenders to capture the attackers under time-optimal control, {as well as} 2) helps eliminate or delay possible future collisions among the defenders on the optimal trajectories. To prevent inevitable collisions on optimal trajectories or collisions arising due to {time-sub-optimal} behavior by the attackers, a minimally augmented control using exponential control barrier function (ECBF) is also provided. {Simulations show the efficacy of the approach.} 
		
		
	\end{abstract}
	
	\section{Introduction}
	

\vishnu{In area defense games, various cooperative techniques are used in 
\xinyi{ \cite{isaacs1999differential,huang2011differential, pierson2016intercepting,chen2017multiplayer,shishika2020cooperative,coon2017control,zheng2020time}} for teams of defenders to intercept multiple attackers before they reach a safety-critical (protected) area.} \vishnu{These approaches typically rely on simplified kinematic motion models, where time-minimizing or distance-minimizing defender-to-attacker, linear assignment is sufficient to ensure collision-free optimal trajectories for the players. However, in real-world applications, the robots operate under the influence of acceleration requiring dynamics motion models, which gives rise to possible collisions on the agents' optimal trajectories. In this paper, we address this critical gap by developing a collaborative defense strategy that explicitly accounts for defender dynamics and the potential for collisions along their optimal interception trajectories. Our approach enables an assignment of defenders to attackers in a way that not only maximizes interception effectiveness but also ensures safe coordination among the defenders.}.

\vishnu{
Numerous pursuit-evasion variants \cite{weintraub2020introduction} have been studied with diverse scenarios and solution methods, including HJI equations \cite{ho1965differential}, isochrones \cite{oyler2016pursuit}, Voronoi-based strategies \cite{pan2012pursuit,pierson2016intercepting}, and numerical schemes \cite{bardi1999numerical}. However, these works often overlook the defense of a specific area, making them less suitable for addressing area protection against attackers.
}

{
Different from standard pursuit-evasion games, the area or target defense problem involves an area (target or protected area) which an attacker player attacks or tries to reach and a defender player aims to prevent the attack.
These area defense games are modeled as zero-sum differential game and solved using various solution techniques including optimal control 
\xinyi{\cite{isaacs1999differential, mohanan2020target}, reachability analysis \cite{huang2011differential, hu2023multi}, semi-definite programming \cite{landry2018reach}, {second-order} cone programming \cite{lorenzetti2018reach}, model predictive control \cite{lee2016model}, control barrier functions (CBFs) \cite{guerrero2021area}, and geometric approach \cite{yan2021cooperative}.}
}


	{Due to the curse of dimensionality, these approaches for 1 attacker vs 1 defender (1-vs-1) area defense games are extended to multi-player scenarios using a ``divide and conquer'' approach. 
	 For example, in \cite{chen2017multiplayer}, the authors solve the reach-avoid game for each pair of defenders and attackers operating in a compact domain with obstacles using a Hamilton-Jacobi-Issacs (HJI) reachability approach. The solution is then used to assign defenders against the attackers in multiplayer case using graph-theoretic maximum matching.} {The authors in \cite{pierson2016intercepting} develop a distributed algorithm for the cooperative pursuit of multiple evaders by multiple pursuers, using an area-minimization strategy based on a Voronoi tessellation in a bounded convex environment. Each defender targets it nearest evader and cooperate with other neighboring pursuers who share the same evader in order to capture all the evaders in finite time.}
	
{In a different approach \cite{shishika2020cooperative}, where defenders are restricted to move on the perimeter of protected area, local reach-avoid games between small teams of defenders and attackers are first solved, and then the teams of defenders are assigned to capture the oncoming attackers using a polynomial-time algorithm. For a reach-avoid game in convex domains, authors in \cite{yan2019task} provide geometric barriers separating the winning regions for groups of pursuers vs one evader by using analytical methods and assign groups of pursuers against the evaders using a mixed integer program.  In \cite{yan2022matching}, the authors decompose a 3D multi-player reach-avoid game into smaller games involving groups of three or fewer pursuers and only one evader by proposing a potential function and then  solve a sequential matching problem 
to assign these groups of pursuers to capture the evaders. \vishnu{The authors in \cite{yan2024pursuit} investigate a multiplayer reach-avoid differential game in the presence of general polygonal obstacles.}}
	
	
	The aforementioned studies provide useful insights to the area or target defense problem, however, are limited in application due to the use of the single-integrator motion model, and due to the lack of consideration of inter-defender collisions. {While some works do use the double-integrator motion model \cite{coon2017control} or the unicycle motion model \xinyi{\cite{zheng2020time, li2022intelligent}}, these approaches do not consider collisions among the defenders that may happen on the defenders' optimal interception trajectories.}


	{In this paper, we resort to the common ``divide and conquer'' paradigm and provide a collaborative interception strategy for the defenders that takes into account any possible future collisions among the defenders on their optimal interception trajectories at the defender-to-attacker assignment stage.} {For a pairwise (one defender one attacker) target defense scenario, a differential game, particularly zero-sum game, solved using HJI reachability approach to find optimal strategies for both players, is an ideal approach. However, this approach is computationally very demanding even for players with 2 state variables, resulting in 4D game \cite{chen2017multiplayer}. Considering a double integrator motion model for players, resulting in 8D game, would be worse in terms of computation, with collision predictions on such optimal trajectories making it even worse. Thus, to keep the overall computation burden low for the defenders' team, in this work, we build on the time-optimal guidance problem for isotropic rocket \cite{bakolas2014optimal}, which uses a damped double-integrator motion model. We first obtain a time-optimal strategy for a defender to capture a given attacker operating under a similar time-optimal control strategy to capture the protected area}. We then use the times of interception by each defender to capture each attacker, and the times of possible collisions on the defenders' optimal trajectories, to assign the defenders to the attackers. We call this assignment the collision-aware defender-to-attacker assignment (CADAA). Furthermore, we use exponential control barrier functions (ECBF) \cite{nguyen2016exponential, wang2017safe} in a quadratically constrained quadratic program (QCQP) to augment defenders' optimal control actions in order to avoid collision with other fellow defenders when such collisions are unavoidable solely by CADAA. {The overall defense strategy that combines CADAA and ECBF-QCQP is called as inter-defender collision-aware interception strategy (IDCAIS).}
	
	{Note that various strategies have been proposed in the literature for collision avoidance control for autonomous agents \cite{huang2019collision}, but
    we chose a framework similar to the ECBF based quadratic programming technique \cite{nguyen2016exponential, wang2017safe} because this technique allows to generate a collision-avoidance control in the event of imminent collisions with a minimal correction to the nominal control action of an agent. }
	
{This is the first time a defender-attacker assignment framework is presented that considers possible future collisions among the defenders on their optimal interception trajectories. In summary, the major contributions of this paper compared to the prior literature are: (1) time-optimal interception strategy for a defender against an attacker with time-optimal control strategy, {both moving under double integrator motion model}; (2) a mixed-integer quadratic program (MIQP) to find collision-aware assignment CADAA in order to capture as many attackers as possible and as quickly as possible, while preventing or delaying the possible future collisions among the defenders; and 3) a heuristic to quickly find future collision times on defenders' time-optimal trajectories. }

	\vishnu{Organization:} {Section \ref{sec:math_model} provides the problem statement. The interception strategy for the 1-defender-vs-1-attacker area defense game is given in Section \ref{sec:one-vs-one_interception_game}, and that for multiple defenders vs multiple attackers is discussed in Section \ref{sec:many-vs-many_interception_game}. Simulation results and conclusions are given in Section \ref{sec:simulations} and Section \ref{sec:conclusions}, respectively.}

	\section{Modeling and Problem Statement}\label{sec:math_model}

	\textit{Notation}: 	
	\vishnu{We use bold letters to denote vector quantities and matrices. We use $\norm{\cdot}$ to denote the Euclidean norm and $\abs{\cdot}$ for absolute value.}	
	 A ball of radius $\rho$ centered at the origin is defined as $\calB_{\rho}=\{\mathbf{r} \in \bR^2| \norm{\mathbf{r}} \le \rho \}$. 
	\vishnu{ $A\backslash B$ denote the set of elements in $A$ but not in $B$. $A \times B$ denotes the cartesian product of sets A and B.}
	  We define $C(\theta) =\cos(\theta)$, {$S(\theta) =\sin(\theta)$} and \vishnu{use $\mathbf{0}_n$ and $\mathbf{I}_n$ to denote n-dimensional zero and identity matrix, respectively.}

	We consider $N_a$ attackers denoted as $\calA_i$, $i \in I_a= \{1,2,...,N_a\}$, and $N_d$ defenders denoted as $\calD_j$, $j \in I_d= \{1,2,...,N_d\}$, operating in a 2D environment $\calW \subseteq \mathbb{R}^2$ that contains {a circular protected area} $\calP \subset \calW$, defined as $\calP=\{\textbf{r} \in \bR^2 \;| \; \norm{\textbf{r}-\textbf{r}_p}\le \rho_p\}$. The number of defenders is no less than that of attackers, i.e., $N_d \ge N_a$. {The agents $\calA_i$ and $\calD_j$ are modeled as point masses.}
	\begin{wrapfigure}[10]{r}{0.24\textwidth}
		\includegraphics[width=1\linewidth,trim={5.5cm 3.15cm 4.5cm 3.75cm},clip]{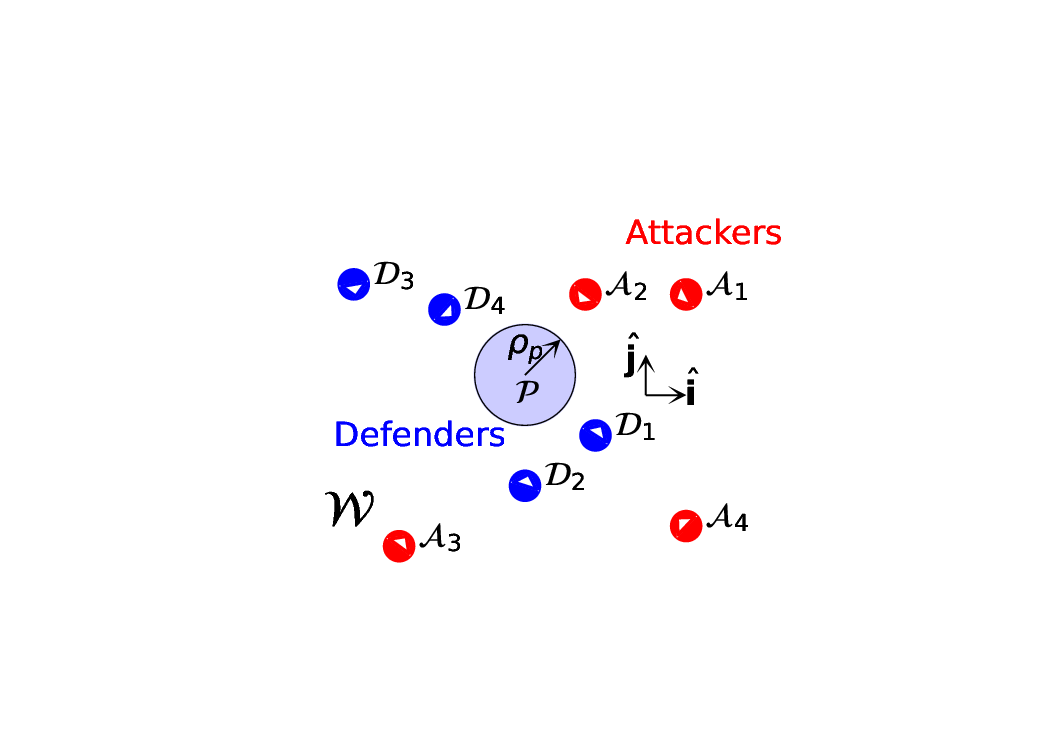}
		\vspace{-6mm}
		\caption{Problem Setup}
		\label{fig:problem_formulation}
	\end{wrapfigure}    

	\vishnu{Let $\mathbf{r}_{ai}(t)=[x_{ai}(t),\; y_{ai}(t)]^T \in \bR^2$ and $\mathbf{r}_{dj}(t)=[x_{dj}(t),\; y_{dj}(t)]^T \in \bR^2$ be the position vectors of $\calA_i$ and $\calD_j$, respectively; $\mathbf{v}_{ai}(t)=[v_{x_{ai}}(t),\; v_{y_{ai}}(t)]^T \in \bR^2$, $\mathbf{v}_{dj}(t)=[v_{x_{dj}}(t),\; v_{y_{dj}}]^T \in \bR^2$ be the velocity vectors, respectively, and $\mathbf{u}_{ai}(t)=[u_{x_{ai}}(t),\; u_{y_{ai}}(t)]^T \in \bR^2$, $\mathbf{u}_{dj}(t)=[u_{x_{dj}}(t),\; u_{y_{dj}}(t)]^T \in \bR^2$ be the accelerations}, which serve also as the control inputs, respectively, all resolved in a global inertial frame $\calF_{gi} (\hat {\mathbf{i}}, \hat {\mathbf{j}})$ (see \vishnu{Fig.~}\ref{fig:problem_formulation}).
	The agents move under double integrator dynamics with linear drag (damped double integrator), similar to isotropic rocket \cite{bakolas2014optimal}:
    \vishnu{
	\be\label{eq:dampedDIDyn}
\dot{\mathbf{x}}_{\imath} (t)
= \underbrace{\bbmat \mathbf{0}_{2} & \mathbf{I}_2 \\ \mathbf{0}_{2} & -C_D\mathbf{I}_2 \ebmat}_{\mathbf{A}} \mathbf{x}_{\imath}(t) + \underbrace{\bbmat \mathbf{0}_2\\ \mathbf{I}_2\ebmat}_{\mathbf{B}} \mathbf{u}_{\imath}(t), \hspace{2mm} \mathbf{x}_{\imath}(0) = \mathbf{x}_{\imath0}
\ee
}
where \vishnu{$\mathbf{x}_{\imath}(t) = [\mathbf{r}_{\imath}^T(t), \mathbf{v}_{\imath}^T(t)]^T \in \bR^4$ is the state vector of agent} ${\imath} \in \{ai \;|\; i \in I_a\} \cup \{dj \; | \; j \in I_d \}$ {(with $ai$ denoting the $i^{th}$ attacker $\calA_i$ and $dj$ denoting the $j^{th}$ defender $\calD_j$)}, \vishnu{$\mathbf{A} \in \bR^{4 \times 4}$ is the state matrix, $\mathbf{B} \in \bR^{4 \times 2}$ is the input matrix}, and $C_D>0$ is the known, constant drag coefficient. The accelerations \vishnu{$\mathbf{u}_{ai}(t)$ and $\mathbf{u}_{dj}(t)$ are bounded by $\bar{u}_a $, $\bar{u}_d$ as: $
	\norm{\mathbf{u}_{ai}(t)} \le \bar{u}_a, \quad
	\norm{\mathbf{u}_{dj}(t)} \le \bar{u}_d. 
$}
\vishnu{The drag term in the} damped double integrator \eqref{eq:dampedDIDyn} inherently poses a speed bound on each agent \vishnu{under bounded acceleration}, i.e.,
\vishnu{
\be \label{eq:velocity_bounds}
\norm{\mathbf{v}_{ai}(t)} \le \bar{v}_a=\frac{\bar{u}_a}{C_D}  \text{ and } \norm{\mathbf{v}_{dj}(t)} \le \bar{v}_d=\frac{\bar{u}_d}{C_D},
\ee
}
\vishnu{provided the initial velocities satisfy $\norm{\mathbf{v}_{ai0}
}<\frac{\bar{u}_a}{C_D}$, for $i \in I_a$ and $\norm{\mathbf{v}_{dj0}
}<\frac{\bar{u}_d}{C_D}$, for $j \in I_d$ (proof in Appendix~\ref{append:velocity_bounds}). This eliminates the need for explicit constraint} on the velocity of the agents while designing bounded controllers, as in earlier literature. {So we have \vishnu{$\mathbf{x}_{ai}(t) \in \calX_a \subset \bR^4$}, for all $i \in I_a$, where $\calX_a = \bR^2 \times \calB_{\bar{v}_a}$ and  \vishnu{$\mathbf{x}_{dj}(t) \in \calX_d \subset \bR^4$}, for all $j \in I_d$, where $\calX_d = \bR^2 \times \calB_{\bar{v}_d}$, \vishnu{with $\calB_{\bar{v}_a}$ and $\calB_{\bar{v}_d}$ denoting the balls of radius $\bar{v}_a$ and $\bar{v}_a$, respectively, around the origin}. We denote by $\calX=\calX_a^{N_a}\times \calX_{d}^{N_d}$ the configuration space of all the agents.} We also make the following assumptions:
 \begin{assumption} \label{assum:fast_defenders}
  	The defenders are at least as fast as the attackers, i.e., $\bar{u}_a \le \bar{u}_d$.
 \end{assumption}
 Note that under Assumption~\ref{assum:fast_defenders}, and by virtue of eq.~\eqref{eq:velocity_bounds}, the maximum velocities of the attackers and defenders also satisfy $\bar{v}_a \le \bar{v}_d$.
\begin{assumption} Each player (either defender or attacker) knows the states of all the other players.  
\end{assumption}

  Each defender $\calD_j$ is endowed with an interception radius $\rho_d^{int}$, i.e., defender $\calD_j$ captures attacker $\calA_i$ when \vishnu{$\norm{\mathbf{r}_{dj}(t)-\mathbf{r}_{ai}(t)}<\rho_d^{int}$ for some $t>0$}.
	{The attackers aim to reach the protected area $\calP$ as quickly as possible. 
 The defenders aim to capture as many of these attackers as possible before they enter the protected area, and ensuring they themselves do not collide with each other. Formally, we consider the following problem.}

 {
	\begin{problem}[Collision-Aware Multi-player Defense] \label{prob:collision_aware_multi_player_defense}
	Design a control strategy $\mathbf{u}_{dj}$ for the defenders $\calD_j$, $\forall j \in I_d$ to: 1) intercept as many attackers and as quickly as possible before they reach the protected area, 2) ensure that the defenders do not collide with each other, i.e., $\norm{\mathbf{r}_{dj}(t)-\mathbf{r}_{dj'}(t)}>\rho_{d}^{col}$ for all $j, j' \neq j \in I_d$ for all $t\ge 0$.
	\end{problem}
 }
Before discussing the solution to Problem \ref{prob:collision_aware_multi_player_defense}, we discuss the one defender against one attacker game as follows.

\section{Interception strategy: one defender against one attacker}\label{sec:one-vs-one_interception_game}
In this section, we first consider a scenario with one attacker attacking the protected area $\calP$ and one defender trying to defend $\calP$ by intercepting the attacker (1D-vs-1A game). The solution to these one-on-one games is then used later to solve Problem \ref{prob:collision_aware_multi_player_defense}. We consider a time-optimal control strategy for the defender. Before we discuss this strategy, we first discuss time-optimal control for an agent moving under  \eqref{eq:dampedDIDyn} in the following subsection.

\subsubsection{Time-optimal control under damped double integrator} \label{sec:timeOptimalControl}
{The time-optimal control problem for an agent starting at $\mathbf{x}_{\imath 0}=[x_{\imath 0},y_{\imath 0},v_{x_{\imath 0}},v_{y_{\imath 0}}]^T$ and moving under \eqref{eq:dampedDIDyn} to reach the boundary of the circular area centered at $\mathbf{r}_{\dag} =[x_{\dag}, y_{\dag}]^T$ with radius $\rho_{\dag}$, is formally defined as: }
\bse \label{eq:minTimeProblem}
\begin{align}
\mathbf{u}_{\imath}^*(t) =\ds \argmin_{\mathbf{u}_{\imath}(t)} \hspace{2mm} & J(\mathbf{u}(t)) =\int_{0}^{t_f} dt\\
\text{subject to  } & 1)\; \dot{\mathbf{x}}_{\imath}(t) = {\mathbf{A}} \mathbf{x}_{\imath}(t)+{\mathbf{B}} \mathbf{u}_{\imath}(t),	\\
&2)\; \norm{\mathbf{u}_{\imath}(t)}\le \bar{u},\\
& 3)\; \mathbf{x}_{\imath}(0)=\mathbf{x}_{\imath 0},\\
& {4)\; \norm{\mathbf{r}_{\imath}(t_f) - \mathbf{r}_{\dag}}^2 - \rho_{\dag}^2 = 0,} \label{eq:MinTimeProblem_constraint4}
\end{align}
\ese
where {$\bar{u}$ is the maximum limit on the acceleration}, {$t_f$ is free terminal time at which the agent reaches the boundary of the circular area.}
\vishnu{Using} Pontryagin's minimum principle \cite{kirk2004optimal}, the optimal control $\mathbf{u}_{\imath}^*(t)$ solving \eqref{eq:minTimeProblem} is given by:
\be \label{eq:timeOptimalControl1}
\baa
\mathbf{u}_{\imath}^*(t) = \bar{u} \big[\frac{-\lambda_3^*(t)}{\sqrt{(\lambda_3^*(t))^2+(\lambda_4^*(t))^2}},  \frac{-\lambda_4^*(t)}{\sqrt{(\lambda_3^*(t))^2+(\lambda_4^*(t))^2}}  \big]^T
\eaa
\ee
where $\bm{\lambda}^*(t) 
= [\lambda_1^*(t),\; 
\lambda_2^*(t),\;\lambda_3^*(t),\;\lambda_4^*(t)]^T= \bbmat c_1 , \hspace{2.5mm} c_2 , \hspace{2.5mm} c_1-c_3e^{-C_D t}, \hspace{2.5mm} c_2-c_4e^{-C_D t} \ebmat^T$ is the co-state vector, {where $c_i$ for $i \in \{1,2,3,4\}$ are some constants that depend on the boundary conditions in \eqref{eq:minTimeProblem}}.
 Similar to the time-optimal control for simple double integrator dynamics \cite{coon2017control,feng1986acceleration}, the optimal control $\mathbf{u}_{\imath}^*(t)$ in \eqref{eq:timeOptimalControl1} is a constant vector for all $t \in [0,t_f]$ as shown in Proposition 3.2 in \cite{bakolas2014optimal}. One can write $\mathbf{u}_{\imath}^*(t)=\bar{u}[C(\theta_{\imath}^*),\; S(\theta_{\imath}^*)]^T$, where $\theta_{\imath}^*$ is the angle made by the vector $\mathbf{u}_{\imath}^*(t)$ with the $x$-axis.
 
{After integrating the system dynamics \eqref{eq:dampedDIDyn} under the constant input $\mathbf{u}_{\imath}^*(t)$, the position trajectories are}:
\be \label{eq:optimal_trajectories}
\baa
x_{\imath}^*(t)=x_{\imath 0}+ E_1(t) v_{x_{\imath 0}} + E_2(t) \bar{u}C(\theta_{\imath}^*),\\
y_{\imath}^*(t)=y_{\imath 0}+ E_1(t) v_{y_{\imath 0}} + E_2(t)  \bar{u}S(\theta_{\imath}^*).
\eaa
\ee
where $E_1(t)=\frac{1-e^{-C_Dt}}{C_D}$ and $E_2(t) =\frac{t-E_1(t)}{C_D} $.
{By using the terminal constraint from \eqref{eq:MinTimeProblem_constraint4} in eq.~\eqref{eq:optimal_trajectories}}, the optimal control $\mathbf{u}_{\imath}^*(t)=\bar{u}[C(\theta_{\imath}^*),\; S(\theta_{\imath}^*)]^T$ and the terminal time $t_f$ can be found by solving the following equations simultaneously \cite{bakolas2014optimal}.
{
\be \label{eq:optimal_control_soln}
\baa
x_{\dag} - \rho_{\dag}C(\theta_{\imath}^*)=x_{\imath 0}+ E_1(t_f) v_{x_{\imath 0}} + E_2(t_f)  \bar{u}C(\theta_{\imath}^*),\\
y_{\dag} - \rho_{\dag}S(\theta_{\imath}^*)=y_{\imath 0}+ E_1(t_f) v_{y_{\imath 0}} + E_2(t_f)  \bar{u}S(\theta_{\imath}^*).
\eaa
\ee
}
\vishnu{The solution to Eq.~\eqref{eq:optimal_control_soln} always exists and is unique. For detailed proof see Appendix~\ref{append:existence_and_uniqueness}.}
\subsubsection{Defender's interception strategy against an attacker}

{In this work, we assume that an attacker is risk-taking, i.e., it does not care about its own survival, and its main goal is to damage the protected area as quickly as possible.  Hence the attacker is assumed to apply time-optimal control calculated in \eqref{eq:attackMinTimeProblem} to reach boundary of the protected area as quickly as possible.
}
\bse \label{eq:attackMinTimeProblem}
\begin{align}
\mathbf{u}_a^*(t)=\argmin_{\mathbf{u}_a(t)} \hspace{2mm} & J_a(\mathbf{u}_a(t)) =\int_{0}^{t_{a}^{int}} dt\\
\text{subject to  } & 1)\; \dot{\mathbf{x}}_a = {\mathbf{A}} \mathbf{x}_a+{\mathbf{B}} \mathbf{u}_a,	\\
& 2)\; \norm{\mathbf{u}_a(t)}\le \bar{u}_a,\\
& 3)\; \mathbf{x}_a(0)=\mathbf{x}_{a0},\\
& 4)\;  \norm{\mathbf{r}_a(t_{a}^{int})-\mathbf{r}_p}^2-(\varrho_p)^2=0,
\end{align}
\ese
{where $t_{a}^{int}$ is the free time that the attacker requires to reach the protected area $\calP$ and $\varrho_p = \rho_p + \rho_{a}^{int}$, where $\rho_{a}^{int}>0$ is the capture radius of the attacker.} 
The defender also employs a time-optimal control action based on the time-optimal control that the attacker may take to reach the protected area. {The control strategy is obtained by solving the following time-optimal problem for the defenders starting at $\mathbf{x}_{d0}$.}
\bse \label{eq:defendBest-against-worstMinTimeProblem}
\noindent
\begin{align}
\mathbf{u}_d^*(t)=\ds \argmin_{\mathbf{u}_d(t)} \hspace{2mm} & J_d(\mathbf{u}_d(t)) = {\int_{0}^{t_{d}^{int}}} dt\\
\text{subject to  } 1) & \; \dot{\mathbf{e}}(t) = {\mathbf{A}} \mathbf{e}(t)+{\mathbf{B}} (\mathbf{u}_d(t) - \mathbf{u}_a^*(t)),	\\
 2) &\;  \mathbf{u}_a^*(t) \text{ is the solution of \eqref{eq:attackMinTimeProblem}},\\
3)& \;\norm{\mathbf{u}_d(t)}\le {\bar{u}_d'},\\
4)& \; \mathbf{e}(0)=\mathbf{e}_{0},\\
5)& \; \vishnu{\norm{\mathbf{e}_{\mathbf{r}}(t_d^{int})}=0}, \label{eq:eq:defendBest-against-worstMinTimeProblem_constraint5}
\end{align}
\ese
where $\mathbf{e}(t)=[e_x(t),e_y(t),e_{v_x}(t),e_{v_y}(t)]^T=\mathbf{x}_d(t)-\mathbf{x}_a(t)$, \vishnu{$\mathbf{e}_{\mathbf{r}}(t) = [e_x(t), e_y(t)]^T$, and} $t_{d}^{int}$ is the time when the defender captures the attacker {and $\bar{u}_d' = \bar{u}_d -\epsilon_d$ where $\epsilon_d$ ($\ll$ 1) is a very small, strictly positive, user-defined constant to ensure that the constraints satisfy the strict complementary slackness condition, and in turn ensure continuity of the solution to the quadratically-constrained quadratic program (QCQP) used for collision avoidance in multi-defender case discussed later in Subsection \ref{sec:collision_avoidance_control}}. \vishnu{Note that we set $\rho_{d}^{int}$ to be 0 in Eq.~\eqref{eq:eq:defendBest-against-worstMinTimeProblem_constraint5} instead of a circle of radius $\rho_{d}^{int}$ (the interception radius), to ensure that the defender directly targets the attacker and maximizes its interception impact; in practice the capture is still considered successful when the defender reaches within $\rho_{d}^{int}$ distance of the attacker. }

Similar to \eqref{eq:timeOptimalControl1}, one can establish that the optimal controls $\mathbf{u}_d^*(t)={\bar{u}_d'}[C(\theta_d^*),S(\theta_d^*)]^T$ and $\mathbf{u}_a^*(t)=\bar{u}_a[C(\theta_a^*),S(\theta_a^*)]^T$ are constant vectors. The angles $\theta_d^*, \; \theta_a^*,$ and the terminal times $ t_{d}^{int}, \; t_{a}^{int}$ can be found by solving the following set of algebraic equations:
{
\be \label{eq:defendBest-against-worstMinTimeSolution}
\begin{aligned}
x_p - \varrho_p C(\theta_a^*)  &\hspace{-0mm}=  x_{a}+ E_1(t_a^{int}) v_{x_{a}} + E_2(t_{a}^{int})  \bar{u}_aC(\theta_a^*),\\
y_p - \varrho_p S(\theta_a^*) &\hspace{-0mm}= y_{a}+ E_1(t_a^{int}) v_{y_{a}} + E_2(t_{a}^{int})  \bar{u}_aS(\theta_a^*) ,\\
0 &\hspace{-0mm}=e_x + E_1(t_d^{int}) e_{v_x} + E_2(t_{d}^{int})  \Delta u_x^*,\\
0 &\hspace{-0mm}=e_{y}+ E_1(t_d^{int}) e_{v_y} + E_2(t_{d}^{int}) \Delta u_y^*.
\end{aligned}
\ee
where $[\Delta u_x^*, \Delta u_y^*]^T = \mathbf{u}_d^*(t) - \mathbf{u}_a^*(t) $.}
 {The defender computes its time-optimal strategy at every time step in order to correct its action in reaction to any {time-sub-optimal} behavior by the attacker. We notice that the control action $\mathbf{u}_d^*(t)$ obtained in \eqref{eq:defendBest-against-worstMinTimeProblem} is a function of the current states of the players $\mathbf{x}_d$ and $\mathbf{x}_a$, i.e., it is a closed-loop control law.}  Next, we discuss winning regions of the players.
 
\subsubsection{Winning regions of the players}
In this section, for a given initial condition $\mathbf{x}_d$ of the defender, we characterize the set of initial conditions of the attacker for which the defender is able to capture the attacker under the control strategy \eqref{eq:defendBest-against-worstMinTimeProblem}. This set, called as winning region of the defender, is denoted by $\calR_d(\mathbf{x}_{d})$:
{
\be \label{eq:defender_winning_region}
\calR_d(\mathbf{x}_{d})= \{\mathbf{x}_a \in  \bar{\calX}_{a}|t_{d}^{int}(\mathbf{x}_{d},\mathbf{x}_{a})-t_{a}^{int}(\mathbf{x}_{a},\mathbf{r}_p) < 0  \},
\ee
where $\bar{\calX}_{a}=(\bR^2\backslash\calP)\times\calB_{\bar{v}_a}$, $t_{a}^{int}(\mathbf{x}_{a},\mathbf{r}_p)$ is the time that the attacker starting at $\mathbf{x}_{a}$ requires to reach the {protected area $\calP$ centered at  $\mathbf{r}_p$}, 
and $t_{d}^{int}(\mathbf{x}_{d},\mathbf{x}_{a})$ is the time that defender starting at $\mathbf{x}_{d}$ requires to capture the attacker starting at $\mathbf{x}_{a}$ under the control strategy in \eqref{eq:defendBest-against-worstMinTimeSolution}.}
Similarly, the winning region of the attacker, denoted as $\calR_{a}(\mathbf{x}_d)$, is defined as
{
$
\calR_a(\mathbf{x}_{d})= \bar{\calX}_{a} \backslash \calR_d(\mathbf{x}_{d})
$.}

{In the following theorem, we formally prove the effectiveness of the defender's time-optimal strategy.
\begin{theorem}
The defender $\calD$, starting at $\mathbf{x}_d$, moving under the control strategy $\mathbf{u}_d^*(t)$ (from \eqref{eq:defendBest-against-worstMinTimeProblem}) against the attacker $\calA$ initially {located in the winning region $\calR_d(\mathbf{x}_{d})$ of the defender given in \eqref{eq:defender_winning_region}}, captures $\calA$ before $\calA$ reaches the protected area $\calP$, i.e., $\norm{\mathbf{r}_d(t_{d}^{int})-\mathbf{r}_a(t_{d}^{int})}$ $\le \rho_{d}^{int}$ and $\norm{\mathbf{r}_p-\mathbf{r}_a(t_{d}^{int})}> \varrho_p$, if $\calA$ {applies time-optimal control} \eqref{eq:attackMinTimeProblem}; 
\end{theorem}
\begin{proof}
Time-optimal trajectories by construction ensure that the defender captures the attacker whenever the attacker starts inside $\calR_d(\mathbf{x}_d)$ before the attacker reaches the protected area.
\end{proof}
}


\section{Collision-Aware Interception Strategy: multiple defenders against multiple attackers}\label{sec:many-vs-many_interception_game}
To solve Problem \ref{prob:collision_aware_multi_player_defense}, we design the defenders' control strategy by:
1) solving each pairwise game, and 2) assigning the defenders to go against particular attackers based on a cost metric that relies on the solutions to the pairwise games, which is a common approach that many of previous works used. However, different from the earlier works, we define the assignment metric based on 1) the time-to-capture and 2) the time of possible collision between the defenders. We aim to find the assignment so that the defenders {1) intercept as many attackers as possible before they reach the protected area, and 2) ensure that defenders have minimum possible collisions among themselves.} We call this collision-aware defender-to-attacker assignment (CADAA). 

{{In the rest of the section, we present the components of CADAA: First, we present a Mixed-Integer Quadratic Program to solve the problem of assigning defenders to attackers based on the time-to-capture and the time of possible inter-defender collisions. Then, we present an algorithm to find time-to-collision for the defenders, as well as a heuristic approach to compute these times in a more efficient manner, which is roughly four times faster than the original algorithm. Finally, we present a controller for  inter-defender collision avoidance that is based on Exponential Control Barrier Functions (ECBFs).}} \footnote{Note that the design of a ECBF-based controller was chosen as a standard control synthesis method that is applicable to constraints of higher-order relative degree with respect to system dynamics, such as the case in this paper, see later on in Section IV-C. Other control methods on collision avoidance could also have been employed.}

\subsection{A MIQP for defender-to-attacker assignment}

 Let $\delta_{ji}$ be the binary decision variable that takes value 1 if the defender $\calD_j$ is assigned to intercept the attacker $\calA_i$ and 0 otherwise. Let  $C_d^{int}(\mathbf{X}_{dj}^{ai})$ be the cost incurred by the defender $\calD_j$ to capture the attacker $\calA_i$ given by:
\be \label{eq:cost_to_capture_attacker}
C_d^{int}(\mathbf{X}_{dj}^{ai})= \begin{cases}
		t_d^{int}(\mathbf{x}_{dj}, \mathbf{x}_{ai}), & \text{if } \mathbf{x}_{ai} \in \calR_d(\mathbf{x}_{dj});\\
		{c_l}, & \text{otherwise};
\end{cases}
\ee
where $\mathbf{X}_{dj}^{ai} = [\mathbf{x}_{dj}^T, \mathbf{x}_{ai}^T]^T$, $t_d^{int}(\mathbf{x}_{dj}, \mathbf{x}_{ai})$ is the time required by the defender {initially located} at $\mathbf{x}_{dj}$ to capture the attacker {initially located} at $\mathbf{x}_{ai}$ as obtained by solving \eqref{eq:defendBest-against-worstMinTimeSolution}, {$c_l\;(>>1)$ is a very large number}, and $\calR_d(\mathbf{x}_{dj})$ is the winning region of the defender {initially located} at $\mathbf{x}_{dj}$ defined in \eqref{eq:defender_winning_region}. As defined in \eqref{eq:cost_to_capture_attacker}, the cost $C_d^{int}(\mathbf{X}_{dj}^{ai})$ {is set to very the large value $c_l$} whenever the defender is unable to capture the attacker before the attacker can reach the protected area. 

Two defenders $\calD_j$ and $\calD_{j'}$ are said to have collided with each other, 
if $\norm{\mathbf{r}_{dj}(t) -\mathbf{r}_{dj'}(t)} < \rho_{d}^{col}$ for some $t\ge 0$, where $\rho_{d}^{col}$ is a user defined collision parameter.
Consider that two defenders starting at $\mathbf{x}_{dj}$ and $\mathbf{x}_{dj'}$ and operating under the optimal control actions obtained from \eqref{eq:defendBest-against-worstMinTimeProblem} that correspond to the attackers starting at $\mathbf{x}_{ai}$ and $\mathbf{x}_{ai'}$, respectively, collide; then, let $t_d^{col}(\mathbf{X}_{dj}^{ai}, \mathbf{X}_{dj'}^{ai'})$ denote the smallest time at which this collision occurs. Since the analytical expressions for optimal trajectories are known, see \eqref{eq:defendBest-against-worstMinTimeSolution}, we can easily find this time $t_d^{col}(\mathbf{X}_{dj}^{ai}, \mathbf{X}_{dj'}^{ai'})$ \vishnu{using the distance evolution between the defenders}.
Let $C_d^{col}(\mathbf{X}_{dj}^{ai}, \mathbf{X}_{dj'}^{ai'})$ \vishnu{be the cost of a potential collision between $\calD_j$ and $\calD_{j'}$}:
\vishnu{
\be
\baa
C_d^{col}(\mathbf{X}_{dj}^{ai}, \mathbf{X}_{dj'}^{ai'}) = 
\frac{1}{t_d^{col}(\mathbf{X}_{dj}^{ai}, \mathbf{X}_{dj'}^{ai'})} - \frac{1}{c_l}
\eaa
\ee
}
{This cost is chosen so that collisions that occur earlier are penalized more in order to delay or, if possible, avoid them.} 
The collision time $t_d^{col}(\mathbf{X}_{dj}^{ai}, \mathbf{X}_{dj'}^{ai'})$ is obtained using Algorithm~\ref{alg:collision_time} and it is assumed that
 $t_d^{col}(\mathbf{X}_{dj}^{ai}, \mathbf{X}_{dj'}^{ai'}) > t_{\epsilon}$, for all $j,j' \in I_d$ where $t_{\epsilon}$ is a small positive number.
 \vishnu{Under dynamics in Eq.~\eqref{eq:dampedDIDyn} and constant acceleration input, the time evolution of the distance between any pair of defenders changes direction at max four times (proof in Appendix~\ref{append:finite_direction_changes_of_distance}).
Thus, in Algorithm~\ref{alg:collision_time}, the function \collisionTime $(\mathbf{x}_{dj}^{ai}, \mathbf{x}_{dj'}^{ai'})$ finds the time-of-collision between $\calD_j$ and $\calD_{j'}$ by finding local optima of distance function by alternatively minimizing the distance and the negative of the distance function (targeting maximum) with initial guess progressively shifted right of the latest optimum over the time interval $[0,t_{max}]$, where $t_{max} = \max(t_d^{int}(\mathbf{x}_{dj}, \mathbf{x}_{ai}),t_d^{int}(\mathbf{x}_{dj'}, \mathbf{x}_{ai'}))$. And, $fmin(f(t),t_0, [\underline{t}, \bar{t}])$  finds the minimum of function $f(t)$ subject to the bounds $[\underline{t}, \bar{t}]$ with $t_0$ being the initial guess (e.g., \textit{fmincon} in MATLAB). Algorithm \ref{alg:collision_time} is guaranteed to find the collision time, if the corresponding defenders collide within the time $[0, t_{max}]$.
}
\begin{algorithm}[ht]
		\caption{Algorithm to find collision times}
		\label{alg:collision_time}
		
		\KwIn{$\mathbf{X}_{d}(0)$,  $\mathbf{X}_{a}(0)$ }
		\KwOut{$\{ t_d^{col}(\mathbf{x}_{dj}^{ai}, \mathbf{x}_{dj'}^{ai'} ) | j, j' \in I_d; i,i' \in I_a \}$}
		\SetKwFunction{collisionTime}{collisionTime}
		\SetKwProg{Fn}{Function}{:}{}
		\For{$j, j' \in I_d ; i, i' \in I_a $}{
            \vishnu{$t_d^{col}( \mathbf{x}_{dj}^{ai}, \mathbf{x}_{dj'}^{ai'} ) = c_l;$}\\
            \If{\vishnu{$j \neq j' \;\&\; i \neq i'$}}{
		$t_d^{col}( \mathbf{x}_{dj}^{ai}, \mathbf{x}_{dj'}^{ai'} )$ =         \collisionTime $(\mathbf{x}_{dj}^{ai}, \mathbf{x}_{dj'}^{ai'})$;
		}
            }
		\Return $\{ t_d^{col}(\mathbf{x}_{dj}^{ai}, \mathbf{x}_{dj'}^{ai'} ) | j, j' \in I_d; i,i' \in I_a \}$
		\vspace{2mm}
		\hrule
		\vspace{1mm}
    	\Fn{\collisionTime{$\mathbf{x}_{dj}^{ai}, \mathbf{x}_{dj'}^{ai'}$}}
    	{
    		$t_0 =0$;
    		\vishnu{$i= 0$}; 
    		$\gamma_t = 0.01$;
    		\vishnu{$\epsilon^{*}= 0.01$};
    		\vishnu{$\epsilon_{m}^{*} = 0.01$};\\
    		\vishnu{$\Delta t^{*} = 2\epsilon^{*}$};
    		\vishnu{$t_{prev}^{*} = t_0$};
    		\\
    		$t_{max} = \max(t_d^{int}(\mathbf{x}_{dj}, \mathbf{x}_{ai}), t_d^{int}(\mathbf{x}_{dj'}, \mathbf{x}_{ai'}))$;\\
            {sign = - sgn($(\mathbf{r}_{dj} -\mathbf{r}_{dj'})^T(\mathbf{v}_{dj} -\mathbf{v}_{dj'})$)};\\
    		 \While {\vishnu{$\bigl (\Delta t^{*}>\epsilon^{*} \;\&\; |t_{max}-t_{prev}^{*}|>\epsilon_{m}^{*} \;\&\;  i<4 \bigr )$}}
    		 {
        		 \text{define} $d(t) = \norm{\mathbf{r}_{dj}(t)-\mathbf{r}_{dj'}(t)}$;\\
                 $[t^{*}, d(t^{*})]$ = \vishnu{$fmin(sign*d(t), t_0, [0,t_{max}]$)};\\
                 \If{\vishnu{$d(t^*)<\rho_{d}^{col}$}}
                 { \Return $t^{*}$;
                 }
                $t_0 = t^{*} + \gamma_t$;\\
                \vishnu{$\Delta t^{*} = |t^{*}-t_{prev}^{*}|$};\\
                $\vishnu{t_{prev}^{*}} = t^{*}$;\\
                \vishnu{$i = i +1$};\\
                $sign = -sign$;
            }
            \Return \vishnu{$c_l$};
	    }
\end{algorithm}

We formulate the following MIQP to assign the defenders to the attackers such that the attackers are captured as quickly as possible while the inter-defender collisions are avoided or delayed as much as possible. 
\bse \label{eq:defender_attackers_assign_MIQP}
	\begin{align}
	\argmin_{\bm{\delta}} & \textstyle \sum_{i \in I_a} \sum_{j \in I_d} \Bigl (({1-w})C_d^{int}(\mathbf{X}_{dj}^{ai})\delta_{ji} + \nonumber\\
	&\quad {w} \textstyle \sum_{i' \in I_a} \sum_{j' \in I_d} C_d^{col}(\mathbf{X}_{dj}^{ai}, \mathbf{X}_{dj'}^{ai'})\delta_{ji}\delta_{j'i'}\Bigr) \label{eq:MILP_cost}\\
	\text{subject to } & \vishnu{\textstyle \sum_{i \in I_a} \delta_{ji} \le 1}, \quad \forall j \in I_d; \label{eq:MILP_constraint_1}\\
    & \textstyle \sum_{j \in I_{d}} \delta_{ji}=1, \quad \forall i \in I_a; \label{eq:MILP_constraint_2}\\
	& \textstyle \delta_{ji} \in \{0,1\}, \quad \forall j \in I_d, \; i \in I_a;
	\end{align}
	\ese
	where $\bm{\delta}=[\delta_{ji}| i \in I_a, j \in I_d]^T \in \{0,1\}^{N_dN_a}$ is the binary decision vector.
	 {The cost in \eqref{eq:defender_attackers_assign_MIQP} is the weighted summation} of the times that the defenders take to capture their assigned attackers, and the total cost associated with possible collisions between the defenders under the optimal control inputs corresponding to the assigned attackers; { and $w \in (0,1)$ is user specified weight of the collision cost. Larger values of $w$ result in assignments with more importance to the collisions among the defenders. In order to give more importance to the interception of the attackers, $w$ is chosen as $w =  0.25$ using the Rank Order Centroid (ROC) weights method \cite{gunantara2018review} that is commonly used in multi-objective optimization. Problem \eqref{eq:defender_attackers_assign_MIQP} is solved by using the MIP solver, Gurobi \cite{gurobi}. \vishnu{Note that, the cost $C_d^{col}(\mathbf{X}_{dj}^{ai}, \mathbf{X}_{dj'}^{ai'})$ has a finite upper bound. This allows choosing $c_l$ sufficiently large such that for any weight $w \in (0,1)$, the assignment ensures no defender is matched to an attacker they cannot intercept before the attacker reaches the protected area, provided an alternative assignment exists (possibly involving some collision conflicts). This maximizes the total number of attackers that can be intercepted. 
 }

\subsection{A heuristic to compute time-to-collision} 

	 The computational complexity of Problem in \eqref{eq:defender_attackers_assign_MIQP} depends on: 1) the computational cost of obtaining the costs $C_d^{int}$ and $C_d^{col}$, and 2) the computational cost of solving the MIQP itself. Solving the MIQP when $C_d^{int}$ and $C_d^{col}$ are given is NP-hard and the worst case computational complexity is $\text{O}(2^{N_{\delta}})$, where $N_{\delta} = N_d N_a $ is the total number of decision variables. {The worst-case computational complexities of finding $C_d^{int}$ and $C_d^{col}$ are $\text{O}(N_{\delta})$ and $\text{O}(N_{\delta}^2)$, respectively, when these costs are obtained in a centralized manner. These computational costs can be reduced further by using distributed computation.} 
	 Additionally, we provide a heuristic in Algorithm  \ref{alg:collision_times_heuristic} that reduces the computation time of finding $C_d^{col}$ using Algorithm~\ref{alg:collision_time}.
	 \begin{algorithm}[ht]
		\caption{A heuristic to find collision times}
		\label{alg:collision_times_heuristic}
		\KwIn{${{\mathbf{X}}}_{d}(0)$,  $\mathbf{X}_{a}(0)$ }
		\KwOut{$\{ t_d^{col}(\mathbf{x}_{dj}^{ai}, \mathbf{x}_{dj'}^{ai'} ) | j, j' \in I_d; i,i' \in I_a \}$}
		\For{$j, j' \in I_d ; i, i' \in I_a $}{ 
		$t_d^{col}( \mathbf{x}_{dj}^{ai}, \mathbf{x}_{dj'}^{ai'} ) = c_l;$\\
		\If{$j \neq j' \;\&\; i \neq i' \; \& \;$ \trianglesIntersect $(\mathbf{x}_{dj}^{ai}, \mathbf{x}_{dj'}^{ai'})$}{
		$t_d^{col}( \mathbf{x}_{dj}^{ai}, \mathbf{x}_{dj'}^{ai'} )$ = \collisionTime $(\mathbf{x}_{dj}^{ai}, \mathbf{x}_{dj'}^{ai'})$
		}
		}
        \vishnu{\Return $\{ t_d^{col}(\mathbf{x}_{dj}^{ai}, \mathbf{x}_{dj'}^{ai'} ) | j, j' \in I_d; i,i' \in I_a \}$}
	\end{algorithm}
	 In Algorithm \ref{alg:collision_times_heuristic}, function
	 \trianglesIntersect
	 $(\mathbf{x}_{dj}^{ai}, \mathbf{x}_{dj'}^{ai'})$ checks if the triangles, that bound the trajectories of $\calD_j$ and $\calD_{j'}$ during the respective time intervals $[0,t_d^{int}(\mathbf{x}_{dj}, \mathbf{x}_{ai})]$ and $[0,t_d^{int}(\mathbf{x}_{dj'}, \mathbf{x}_{ai'})]$, intersect (see \vishnu{Fig.~}\ref{fig:pathCollisionCheck}). {Note that the intersection of the bounding triangles is necessary but not sufficient to say that the corresponding trajectories intersect.}
	 
	  \setlength{\columnsep}{3pt}
   

    \fboxsep=0mm
    \begin{figure}
        \fbox{ \begin{minipage}[c]{0.23\textwidth}
             \includegraphics[width=1\linewidth,trim={5.9cm 3.1cm 1.3cm 1.9cm},clip]{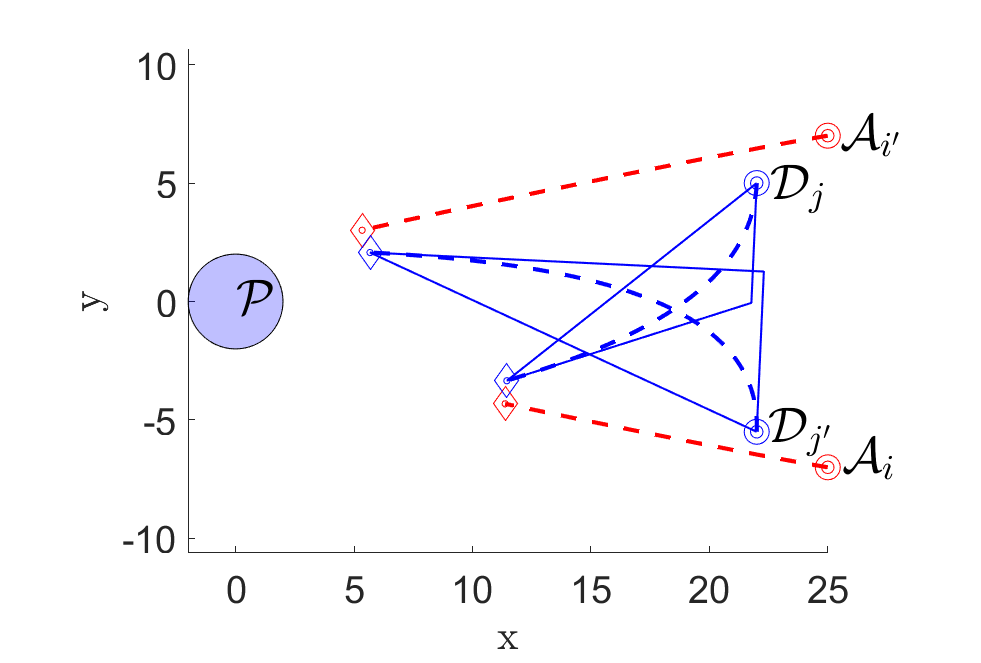}
         \end{minipage}
	\begin{minipage}[c]{0.23\textwidth}
	 	\caption*{The triangle enclosing a defender's trajectory is formed by the line segment joining the start and terminal position and the tangents to the defender's trajectory at the start and the terminal position.}
        \end{minipage}
        }
        \caption{Collision Check via Triangle Intersection. Circles: start positions, diamonds: terminal positions, dotted lines: actual trajectories of the players.}
        \label{fig:pathCollisionCheck}
    \end{figure}
        
	The \collisionTime function requires more time to run than \trianglesIntersect. Hence, using \trianglesIntersect as a qualitative check to decide if indeed collision occurs, and then using \collisionTime to find the time of collision, reduces on average the total computation time.
	{As shown in Table \ref{tab:computation_time_compare}, Algorithm \ref{alg:collision_times_heuristic} on average runs almost four times faster compared to Algorithm~\ref{alg:collision_time} for all possible pairs of trajectories.}
    
    \renewcommand{\arraystretch}{1.3}
 
    \begin{table}[ht]
        \centering
        \begin{tabular}{|c|c|c|c|c|}
        \hline
             $\downarrow$  Algorithm $ \| \; N_a \rightarrow$ & {15}  &  {20}  & {25}   & {30}    \\
             \hline
             {Algorithm~\ref{alg:collision_time}} & 222.1& 760.3 & 1901.4 & 3883.5\\[1mm]
            
             {Algorithm~\ref{alg:collision_times_heuristic}} & 61.0 & 208.3 & 502.3 & 1057.1 \\[1mm]
             \hline
        \end{tabular}
        \caption{ \vishnu{Computation time (s) to find collision times}}	\label{tab:computation_time_compare}
    \end{table}
	
	 The solution to \eqref{eq:defender_attackers_assign_MIQP} always yields an assignment for which defenders do not collide with each other if (i) all the players play optimally as per \eqref{eq:attackMinTimeProblem} and \eqref{eq:defendBest-against-worstMinTimeProblem}, and (ii) such an assignment exists. However, there exist some set of initial states of the players for which such a collision-free assignment does not exist. Furthermore, the attackers may play {time-sub-optimally}, in which case the defenders may collide with each other under the original assignment, irrespective of whether the original assignment was collision-free or not. In such cases, the time-optimal control action of the involved defenders is augmented by a control term that is based on a control barrier function (CBF). We provide a CBF-based quadratically constrained quadratic program (QCQP) to minimally augment the optimal control of the defenders in the following subsection.
	 
	 \subsection{{Control for inter-defender collision avoidance}} \label{sec:collision_avoidance_control}
	 {In this subsection, we consider the inter-defender collision avoidance. We propose a solution based on Exponential Control Barrier Functions (ECBFs) as explained below. We first formulate the QCQP that computes a control input that aims to preserve safety among the defenders while minimally deviating from their nominal control inputs. Then we present conditions under which this QCQP is feasible for all times. Finally we present a theoretical analysis on sufficient conditions that guarantee existence and uniqueness of the solution of the QCQP over a time interval (for the solution to be well-defined).}
	 
	  {Let the function $h_{jj'}(t): \calX \rightarrow \bR$ be defined as $$h_{jj'}(t) = (\rho_d^{col})^2-\norm{\mathbf{r}_{dj}(t)-\mathbf{r}_{dj'}(t)}^2, \; t\geq 0.$$} We require that $h_{jj'}(t) \le 0, \; \forall t\ge 0$ for all $j \neq j' \in I_d$ to ensure that the defenders do not collide with each other.  {In other words, we want the safe set $\calX_{safe}$ to be forward invariant, where $$\calX_{safe} = \{\mathbf{X}\in \calX \;|\; h_{jj'}(t)\le 0, \forall j,j' \in I_d, j\neq j'\},$$ where the combined state vector $\mathbf{X}=[\mathbf{x}_{d1}^T,\mathbf{x}_{d2}^T,...,\mathbf{x}_{dN_d}^T,\mathbf{x}_{a1}^T,\mathbf{x}_{a2}^T,...,\mathbf{x}_{aN_a}^T]^T$.} 
	  
	  {{A sufficient condition in order to render a set forward invariant is given in terms of a Control Barrier Function (CBF), while Quadratic programs (QPs) subject to CBF constraints have been proposed in recent literature to yield real-time implementable control solutions that can be used to augment nominal controllers at each time step for collision avoidance \cite{ames2016control, nguyen2016exponential, xiao2019control, wang2017safe}.}} In our problem, since the function $h_{jj'}(t)$ has relative degree 2 with respect to (w.r.t) the dynamics \eqref{eq:dampedDIDyn}, we resort to the class of exponential CBF (ECBF) \cite{nguyen2016exponential,wang2017safe} for collision avoidance. 
	  
	  For the ease of notation, we drop the argument $t$ in the following text. Let $\Delta \mathbf{u} = [\Delta \mathbf{u}_{d1}^T, \Delta \mathbf{u}_{d2}^T,..., \Delta \mathbf{u}_{dN_d}^T ]^T \in \bR^{1 \times 2N_d}$ be the input correction vector for inter-defender collision avoidance, where $\Delta \mathbf{u}_{dj}$ is the input correction to $\calD_j$, for all $j \in I_d$. 
	  {
	  Following \cite{nguyen2016exponential}, we define the following ECBF that is tailored to the application in this paper.}
	  {
	  \begin{definition}\label{def:ECBF}
	  [Exponential Control Barrier Function]
	  Consider a pair of defenders $\calD_j$ and $\calD_{j'}$, their dynamics given by \eqref{eq:dampedDIDyn} and the set $\calX_{0} = \{\mathbf{X}_{jj'} \in \calX_d^2|\; h_{jj'}\le 0\}$, where $\mathbf{X}_{jj'} =[\mathbf{x}_{dj}^T ,\mathbf{x}_{dj'}^T]^T$ and $h_{jj'}$ has relative degree 2 w.r.t. the system dynamics \eqref{eq:dampedDIDyn}. $h_{jj'}$ is an Exponential Control Barrier Function (ECBF) if there exists $k_{jj'} \in \bR$ such that,
	  \be
	  \inf_{\Delta \mathbf{u}_{dj}, \Delta \mathbf{u}_{dj'} } \ddot{h}_{jj'} + 2k_{jj'}\dot{h}_{jj'} + k_{jj'}^2 h_{jj'}  \le 0, \forall \mathbf{X}_{jj'}\in \calX_0
	  \ee
	  and $h_{jj'} \le C_{jj'} e^{A_{jj'}t}\eta_{jj'} \le 0$ when $\mathbf{X}_{jj'}(0)\in \calX_0$, where $A_{jj'} \in \bR^2$ depends on $k_{jj'}$, and
	  $C_{jj'}=[1,0]$, $\eta_{jj'} = [h_{jj'},\dot{h}_{jj'}]$.
	  \end{definition}
	  }
	 
	    Next, building on the ECBF-QP formulation in \cite{nguyen2016exponential, wang2017safe}, we develop the following ECBF-QCQP as we have quadratic constraint on the control inputs. 
	  {
	 \bse\label{eq:ECBF_QCQP}
	 \begin{align}
	 \ds \min_{\Delta \mathbf{u}} &\quad  \norm{\Delta \mathbf{u}}^2\\
	\text{s.t. } &\; 1)\;   \mathbf{A}_{jj'} \Delta \mathbf{u} - b_{jj'} \le 0, \; \forall j' \in I_d, \quad \forall j < j'; \label{eq:ECBF_QCQP_constraint_1}\\
	&\; 2) \;  \norm{\mathbf{u}_{dj}^*+\Delta \mathbf{u}_{dj}}^2 - (\bar{u}_d)^2 \le 0, \quad \forall j \in I_d; \label{eq:ECBF_QCQP_constraint_2}
	 \end{align}
	 \ese
where $\mathbf{A}_{jj'} = [ 0,...,-2(\mathbf{r}_{dj}-\mathbf{r}_{dj'})^T,...,2(\mathbf{r}_{dj}-\mathbf{r}_{dj'})^T,...,0 ]  \in \bR^{1\times 2N_d}$, $b_{jj'} = -2k_{jj'}\dot{h}_{jj'} - k_{jj'}^2 h_{jj'} + 2 (\norm{\mathbf{v}_{dj}-\mathbf{v}_{dj'}}^2 + (\mathbf{r}_{dj}-\mathbf{r}_{dj'})^T(\mathbf{u}_{dj}^*-C_D \mathbf{v}_{dj} - \mathbf{u}_{dj'}^*+C_D \mathbf{v}_{dj'}))$.} {The constraints in \eqref{eq:ECBF_QCQP_constraint_1} are the safety constraints as established by ECBF (Definition~\ref{def:ECBF})} and those in \eqref{eq:ECBF_QCQP_constraint_2} are quadratic constraints on the control input. 
{The constant coefficients $k_{jj'} > 0, \; \forall j \neq j' \in I_d$, are chosen appropriately to ensure that the above QCQP is feasible.} {Below we explain how to choose the values of $k_{jj'}$}. 

{Let us define $\psi_{jj'}(t)$ as
\be
\psi_{jj'}(t)=\dot{h}_{jj'}(t) + k_{jj'}  h_{jj'}(t).
\ee
{To ensure forward invariance of the safe set of the defenders using ECBFs \cite{nguyen2016exponential,wang2017safe} the initial states of the defenders need to satisfy  Condition~\ref{cond:ECBF_initial_states}.}
\begin{condition} \label{cond:ECBF_initial_states}
1) $h_{jj'}(0) \le 0$, and 2) $\psi_{jj'}(0) \le 0 $, for all $j \neq j' \in I_d$.
\end{condition}
}

\vishnu{To ensure this, we make the following assumption. 
\begin{assumption}[QCQP Feasibility]\label{assum:qcqp_feasibility}
    For the initial states in $S_{jj'}(\rho_0) = \{\mathbf{x}_{dj}(0), \; \mathbf{x}_{dj'}(0) \in \bR^2\times \calB_{\bar{v}_d}| \norm{\mathbf{r}_{dj}(0)- \mathbf{r}_{dj'}(0)}>\rho_0\}$, the parameter $k_{jj'}$ is chosen depending on the initial conditions such that $\frac{4\bar{v}_d\rho_0}{\rho_0^2-(\rho_d^{col})^2}\le k_{jj'}$. Here $\rho_0 = 2\rho_d^b + \rho_d^{col}$, with $\rho_d^b = \frac{\bar{u}_d (1-\log(2))}{C_D^2}$ being the breaking distance traveled by a defender moving with the maximum velocity despite maximum acceleration applied opposite to its velocity and $\rho_d^{col}>0$ being the safety distance between the defenders.
\end{assumption}
Assumption~\ref{assum:qcqp_feasibility} ensures that $\psi_{jj'}(0) \le 0 $ in $S_{jj'}(\rho_0)$ such that Condition~\ref{cond:ECBF_initial_states} holds and QCQP is always feasible for all future times.} 


{
{
Next, we show that the solution to the ECBF-QCQP in \eqref{eq:ECBF_QCQP} is continuously differentiable (hence locally Lipschitz) almost everywhere in the set $\bar{\calX}_{safe}$ defined as
\be \label{eq:x_bar_safe}
\bar{\calX}_{safe} = \big \{\mathbf{X}\in \calX_{safe}| \psi_{jj'}(t) \le 0 ,
 \forall j \neq j' \in I_d \big\}.
\ee} We first present Lemma \ref{lem:measure_zero_set}.
\begin{lemma}\label{lem:measure_zero_set}
Let, without loss of generality, that the attackers $\calA_j$ and $\calA_{j'}$ are assigned to the defender $\calD_{j}$ and $\calD_{j'}$, respectively, and $\calX_2 = \calX_{a}^2 \times \calX_{d}^2$ be the configuration space of two defenders and two attackers. Then, the set $Q_0 = \{[\mathbf{x}_{aj}^T, \mathbf{x}_{aj'}^T, \mathbf{x}_{dj}^T, \mathbf{x}_{dj'}^T]^T \in \calX_2   | \;b_{jj'} = 0, \forall j \neq j' \in I_d \}$ is a measure zero set.
\end{lemma}
\begin{proof}
			To show that $Q_0 \subset \calX_2$ is a measure zero set, we first consider a subset of $\calX_2$ in which we fix the states $\mathbf{x}_{dj}$, $\mathbf{x}_{aj}$, $\mathbf{x}_{aj'}$ and position $\mathbf{r}_{dj'}$ of $\calD_j$, $\calA_j$, $\calA_{j'}$ and $\calD_{j'}$ respectively, and vary the velocity $\mathbf{v}_{dj'}$ of $\calD_{j'}$ such that $\mathbf{v}_{dj'} =v[C(\theta), S(\theta)]^T$ for fixed value of $v$ and different values of $\theta$ (see \vishnu{Fig.~}\ref{fig:measureZeroFig}), i.e., we fix the speed of $\calD_{j'}$ and vary the direction of its velocity vector.

        \fboxsep=0mm
        \begin{figure}
             \includegraphics[width=0.45\textwidth,trim={0cm 2.33cm 0cm 2.1cm},clip]{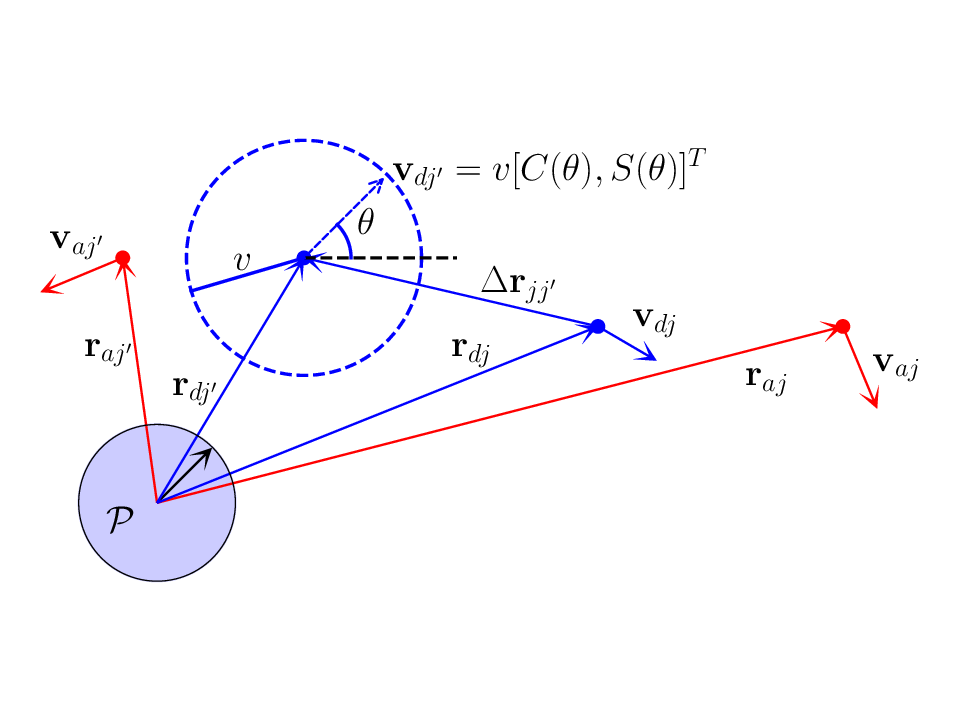}
	 	\caption{\vishnu{Visual of set $\bar{\calX}_2(\mathbf{x}_{aj}, \mathbf{x}_{aj'}, $ $ \mathbf{x}_{dj}, \mathbf{r}_{dj'},v)$ (the dotted circle). Solid vectors denote the fixed vectors and the dotted one denotes the variable vector.}}
        \label{fig:measureZeroFig}       
        \end{figure}
        
			We denote this subset of $\calX_2$ as $\bar{\calX}_2(\mathbf{x}_{aj},\mathbf{x}_{aj'},\mathbf{x}_{dj}, \mathbf{r}_{dj'}, v) = \bar{\calX}_2 = \{\mathbf{v}_{dj'} \in \calB_{\bar{v}_d} | \norm{\mathbf{v}_{dj'}} = v \}$. Now we want to find for what values of $\theta$, we have $b_{jj'} = 0$ on the set $\bar{\calX}_2(\mathbf{x}_{aj},\mathbf{x}_{aj'},\mathbf{x}_{dj}, \mathbf{r}_{dj'}, v)$. On the set $\bar{\calX}_2(\mathbf{x}_{aj},\mathbf{x}_{aj'},\mathbf{x}_{dj}, \mathbf{r}_{dj'}, v)$, we have:
		\be \label{eq:bjj}
		\baa
		b_{jj'} &\hspace{-2.5mm}=  -2k_{jj'}\dot{h}_{jj'} - k_{jj'}^2 h_{jj'} + 2 (\norm{\mathbf{v}_{dj}-\mathbf{v}_{dj'}}^2 +\\ 
		&\hspace{-2.5mm} \quad (\Delta \mathbf{r}_{jj'})^T(\mathbf{u}_{dj}^*-C_D \mathbf{v}_{dj} - \mathbf{u}_{dj'}^*+C_D \mathbf{v}_{dj'}))\\
		&\hspace{-2.5mm}= \underbrace{ - k_{jj'}^2 h_{jj'} + 2(\Delta \mathbf{r}_{jj'}^T(\mathbf{u}_{dj}^*-C_{k_{jj'}} \mathbf{v}_{dj} ) \norm{\mathbf{v}_{dj}}^2+v^2)}_{T_1}\\
		& \hspace{-2.5mm}\quad \underbrace{-2 \big(\Delta \mathbf{r}_{jj'}^T(\bar{u}_d'\hat{\mathbf{o}}(\theta_{dj'}^*)-C_{k_{jj'}} v \hat{\mathbf{o}}(\theta)) + v \mathbf{v}_{dj'}^T\hat{\mathbf{o}}(\theta) \big)}_{T_2} 
		\eaa
		\ee
		where $\Delta \mathbf{r}_{jj'} = \mathbf{r}_{dj}-\mathbf{r}_{dj'}$, $\hat{\mathbf{o}}(\theta) = [C(\theta), S(\theta)]^T$, $C_{k_{jj'}} = C_D-k_{jj'}$ and $\theta_{dj'}^*$ is dependent on $\theta$. On $\bar{\calX}_2(\mathbf{x}_{aj},\mathbf{x}_{aj'},\mathbf{x}_{dj}, \mathbf{r}_{dj'}, v)$ the term $T_2$ in \eqref{eq:bjj} depends on $\theta$ while $T_1$ does not. The term $T_2$ can be further simplified as:
		\be \label{eq:T2}
		\baa
		T_2 & \hspace{-2.5mm}= -2 \big(\Delta \mathbf{r}_{jj'}^T(\bar{u}_d'\hat{\mathbf{o}}(\theta_{dj'}^*)-C_{k_{jj'}} v \hat{\mathbf{o}}(\theta)) + v \mathbf{v}_{dj'}^T\hat{\mathbf{o}}(\theta) \big)\\
		&\hspace{-2.5mm} = -2 \big(\norm{\Delta \mathbf{r}_{jj'}}(\bar{u}_d'\cos(\theta_{ru}+\theta_{dj'}^*)-C_{k_{jj'}} v \cos(\theta_{rv}+\theta)) \\
		&\hspace{-2.5mm} \quad + v \norm{\mathbf{v}_{dj'}}\cos(\theta_{vv}+\theta) \big)
		\eaa
		\ee
		where $\theta_{ru}$,  $\theta_{rv}$,  $\theta_{vv}$ are the angles between vectors $\Delta \mathbf{r}_{jj'}$ and $\mathbf{u}_{dj'}^*$, $\Delta \mathbf{r}_{jj'}$ and $\mathbf{v}_{dj'}$, and vectors $\mathbf{v}_{dj}$ and $\mathbf{v}_{dj'}$, respectively, when $\theta=0$. From \eqref{eq:T2} we observe that $T_2$ is a sum of weighted cosines that depend on $\theta$. {This means $T_2$ can be 0 only for a countably finite number of values of $\theta$ in the set $[0,2\pi]$}. This implies that $b_{jj'}$ can be 0 on $\bar{\calX}_2(\mathbf{x}_{aj},\mathbf{x}_{aj'},\mathbf{x}_{dj}, \mathbf{r}_{dj'}, v)$ only for a countably finite number of values of $\theta$. Let this set of values of $\theta$ be denoted as $\bar{Q}_0$. Then we have that the set $\bar{Q}_0 \subset  \bar{\calX}_2(\mathbf{x}_{aj},\mathbf{x}_{aj'},\mathbf{x}_{dj}, \mathbf{r}_{dj'}, v)$ is a measure zero subset of $\bar{\calX}_2(\mathbf{x}_{aj},\mathbf{x}_{aj'},\mathbf{x}_{dj}, \mathbf{r}_{dj'}, v)$. Next, we have that the configuration space $\calX_2$ of $\calD_j$, $\calD_{j'}$, $\calA_j$, $\calA_{j'}$ is a union of all possible disjoint sets $\bar{\calX}_2(\mathbf{x}_{aj},\mathbf{x}_{aj'},\mathbf{x}_{dj}, \mathbf{r}_{dj'}, v)$. This implies that the set $Q_0$ is exactly the union of all possible subsets $\bar{Q}_0$, i.e., $Q_0= \cup \bar{Q}_0$. Since each $\bar{Q}_0$ is a measure zero subset, we have that the set $Q_0$ is also a measure zero set \cite{bogachev2007measure}. 
\end{proof}
}

{
Next, we make the following assumption.
{
\begin{assumption} \label{assum:num_active_constraints}
The total number of active constraints in \eqref{eq:ECBF_QCQP} at the nominal control $\mathbf{u}_{dj}^*$ (i.e., at $\Delta \mathbf{u}_{dj}=\mathbf{0}$) for all $j \in I_d$,  at any time is less then $2N_d$.
\end{assumption}
Assumption~\ref{assum:num_active_constraints} ensures that the gradients of the active constraints at the nominal control $\mathbf{u}_{dj}^*$, for all $j \in I_d$, are linearly independent. The constraints in \eqref{eq:ECBF_QCQP_constraint_2} can never be active at the nominal control because $\norm{\mathbf{u}_{dj}^*} = \bar{u}_d' = \bar{u}_d-\epsilon_d$. Hence the Assumption~\ref{assum:num_active_constraints} practically implies that only any $N_d$ pairs of the defenders have their corresponding safety constraint in \eqref{eq:ECBF_QCQP_constraint_1} active under the nominal control (i.e., at $\Delta \mathbf{u}_{dj} = \mathbf{0}$). In other words, Assumption~\ref{assum:num_active_constraints} enforces that maximum $N_d$ pairs of defenders could possibly be on a collision course.}

{Next, in Theorem~\ref{thm:qcqp_solution_continuity} we prove that the solution to the ECBF-QCQP \eqref{eq:ECBF_QCQP} is continuously differentiable almost everywhere.} 
		\begin{theorem}\label{thm:qcqp_solution_continuity}
		{
		\vishnu{Let Assumption \ref{assum:qcqp_feasibility} and \ref{assum:num_active_constraints} hold, then} the solution to the ECBF-QCQP in \eqref{eq:ECBF_QCQP} is continuously differentiable (and hence locally Lipschitz) on $\bar{\calX}_{safe} \backslash Q_0$.}
		\end{theorem}
		}
		\begin{proof}
			We observe that: 1) the objective function and the constraints defining the QCQP \eqref{eq:ECBF_QCQP} 
			are twice continuously differentiable functions of $\Delta \mathbf{u}$ and $\mathbf{U}=(\mathbf{x}_{d1},\mathbf{x}_{d2}, ... , \mathbf{x}_{dN_d}, \mathbf{u}_{d1}^*, \mathbf{u}_{d2}^*, ..., \mathbf{u}_{dN_d}^* )$, 2) the gradients of the active constraints with respect to $\Delta \mathbf{u}$ are linearly independent under Assumption \ref{assum:num_active_constraints}, 3) the Hessian of the Lagrangian, $H = \mathbf{I}_{2N_d} + \sum_{l=1}^{N_c} \mu_l \mathbf{I}_{2N_d} $, is positive definite, where $\mathbf{I}_{2N_d}$ is the identity matrix in $\bR^{2N_d \times 2N_d}$ and $\mu_l \ge 0$ are Lagrangian multipliers. Furthermore, since we use $\bar{u}_d' = \bar{u}_d-\epsilon_d$ to obtain the nominal time-optimal control $\mathbf{u}_{dj}^*$, the constraints in \eqref{eq:ECBF_QCQP_constraint_2} will never be active at the nominal control $\mathbf{u}_{dj}^*$ (i.e., at $\Delta \mathbf{u} = \bm{0}$). This implies that the strict complementary slackness condition holds for \eqref{eq:ECBF_QCQP} almost everywhere except for a set $Q_0 = \{[\mathbf{x}_{aj}^T, \mathbf{x}_{aj'}^T, \mathbf{x}_{dj}^T, \mathbf{x}_{dj'}^T]^T \in \calX_2   | ; b_{jj'} = 0, \forall j \neq j' \in I_d \}$. From Lemma \ref{lem:measure_zero_set} we have that $Q_0$ is a measure zero set.
			Then, using the Theorem 2.1 in \cite{fiacco1976sensitivity}, the solution to \eqref{eq:ECBF_QCQP} is shown to be a { continuously differentiable function of $\mathbf{U}$ in $\bar{\calX}_{safe} \backslash Q_0$ and hence the solution is also locally Lipschitz continuous in $\bar{\calX}_{safe} \backslash Q_0$.}

		\end{proof}
		
		{
		At this point, we would like to note that Theorem~\ref{thm:qcqp_solution_continuity} establishes only the local Lipschitz continuity of the solution to ECBF-QCQP \eqref{eq:ECBF_QCQP} on the set $\bar{\calX}_{safe}\backslash Q_0$, and hence the existence and uniqueness of the defenders' trajectories on $\bar{\calX}_{safe}\backslash Q_0$ only up to a maximal time interval $[0,t_{max})$, where $t_{max}>0$ depends on the initial state $\mathbf{X}(0)$ of the players. However, proving forward invariance of the set $\bar{\calX}_{safe}$ using tools such as Nagumo's Theorem \cite{blanchini2008set} (commonly used in CBF community) assumes the existence and uniqueness of the system trajectories \emph{for all} future times. Since the strict complementary slackness may not hold on the set $Q_0$, one can not apply the existence tools in \cite{fiacco1976sensitivity,tam1999continuity} to establish the Lipschitz continuity properties of the solution to the ECBF-QCQP. In addition, the set $Q_0$ may contain system states under which two defenders are in a deadlock, i.e., states that correspond to defenders no longer moving. The conditions that lead to deadlocks are dictated by the motion of the attackers and the corresponding nominal optimal interception control \eqref{eq:defendBest-against-worstMinTimeSolution} of the defenders. Since the attackers' control strategy is not known a priori, commenting on whether/when these types of behavior may occur is again a challenging task that goes beyond the scope of the current paper.
		
		In summary, the focus of the current paper is the collision-aware defender-to-attacker assignment (CADAA) strategy. The ECBF-based collision avoidance control was chosen to complement the collision-aware interception strategy because it allows us to minimally augment the nominal optimal interception control only when needed. As mentioned earlier, other collision avoidance methods can be chosen as well. The safety analysis of the multi-agent system using the ECBF-based collision avoidance control is in general a challenging problem and will be investigated in future work.
		}

\section{Simulation Results}\label{sec:simulations}
In this section, {we provide MATLAB simulations for various scenarios to demonstrate the effectiveness of inter-defender collision-aware interception strategy (IDCAIS), the overall defense strategy that combines CADAA and ECBF-QCQP}. { The proposed CADAA algorithm is compared to a baseline Collision-Unaware Defender-to-Attacker Assignment (CUDAA) algorithm, which is designed using only the time-to-capture term in the optimization cost and solved using the Hungarian Algorithm \cite{kuhn1955hungarian}}. Some key parameters used in the simulations are: $\rho_a=\rho_d=0.5 \hspace{.5mm} m$, $C_D=0.5$, $\bar{v}_a=6 \hspace{.5mm} ms^{-1} \;(\bar{u}_a=3 \hspace{.5mm}  ms^{-2})$, $\bar{v}_d=6.8 \hspace{.5mm}  ms^{-1} \; (\bar{u}_d = 3.4 \hspace{.5mm} ms^{-2})$, $\rho_{d}^{int}=1 \hspace{.5mm} m$, $\rho_{d}^{col}=2 \hspace{.5mm} m$, $\mathbf{r}_p=[0,0]^T$, $\rho_p=2 \hspace{.5mm} m$ {and $w=0.5$}.

{Note that the case studies below  involve small-scale scenarios of up to 4 defenders against the same number of attackers so that all the resulting trajectories are legible in the provided plots, and so that both the cases of (i) attackers playing the time-optimal controller and (ii) attackers not playing the time-optimal controller are clearly illustrated and analyzed. The video of these simulations and additional simulations of up to 16 players on each team (discussed in Section \ref{sec:additional_simulations}) can be found at \href{https://tinyurl.com/yduuruff}{https://tinyurl.com/yduuruff}. }

\subsection{{Attackers use time-optimal control}} 
Consider {Scenario 1} with two defenders (in blue) and two attackers (in red) as shown in \vishnu{Fig.~}\ref{fig:assign_compare}(a) both operating under their time-optimal control actions with no active collision-avoidance control on the defenders. The optimal paths of the players are shown in solid when CADAA is used to assign defenders to the attackers, while for the same initial conditions the dotted paths correspond to the collision-unaware defender-to-attacker assignment (CUDAA). Here CUDAA is obtained based only on the linear term in the objective in \eqref{eq:defender_attackers_assign_MIQP}, which attempts to minimize the sum of the times of interception expected to be taken by the defenders. CUDAA results into $\calD_1 \rightarrow \calA_1$ and $\calD_2 \rightarrow \calA_2$ assignment, while CADAA results into $\calD_1 \rightarrow \calA_2$ and $\calD_2 \rightarrow \calA_1$ assignment. In \vishnu{Fig.~}\ref{fig:assign_compare_dist}(a), the distance between the defenders is shown for CADAA (solid) and CUDAA (dashed). The dotted pink line in \vishnu{Fig.~}\ref{fig:assign_compare_dist}(a) and subsequent figures denote the minimum safety distance $\rho_{d}^{col}$ between the defenders. Note that the defenders collide with each other under CUDAA, which does not account for the possible future collisions among the defenders. On the other hand, CADAA helps avoid the impending collision between the defenders at the cost of delayed interception of the attackers.

In \vishnu{Fig.~}\ref{fig:assign_compare}(b), {Scenario 2} is shown where the CADAA is unable to prevent the collision between the defenders. This is evident from the plot of the distance between the defenders in both cases (\vishnu{Fig.~}\ref{fig:assign_compare_dist}(b)). The CADAA only helps in delaying the collision between the defenders. 

\begin{figure}[ht]
\centering
\begin{subfigure}[ht]{0.95\linewidth}
\includegraphics[width=1\linewidth,trim={1cm 0.45cm 0.7cm 1.2cm},clip]
{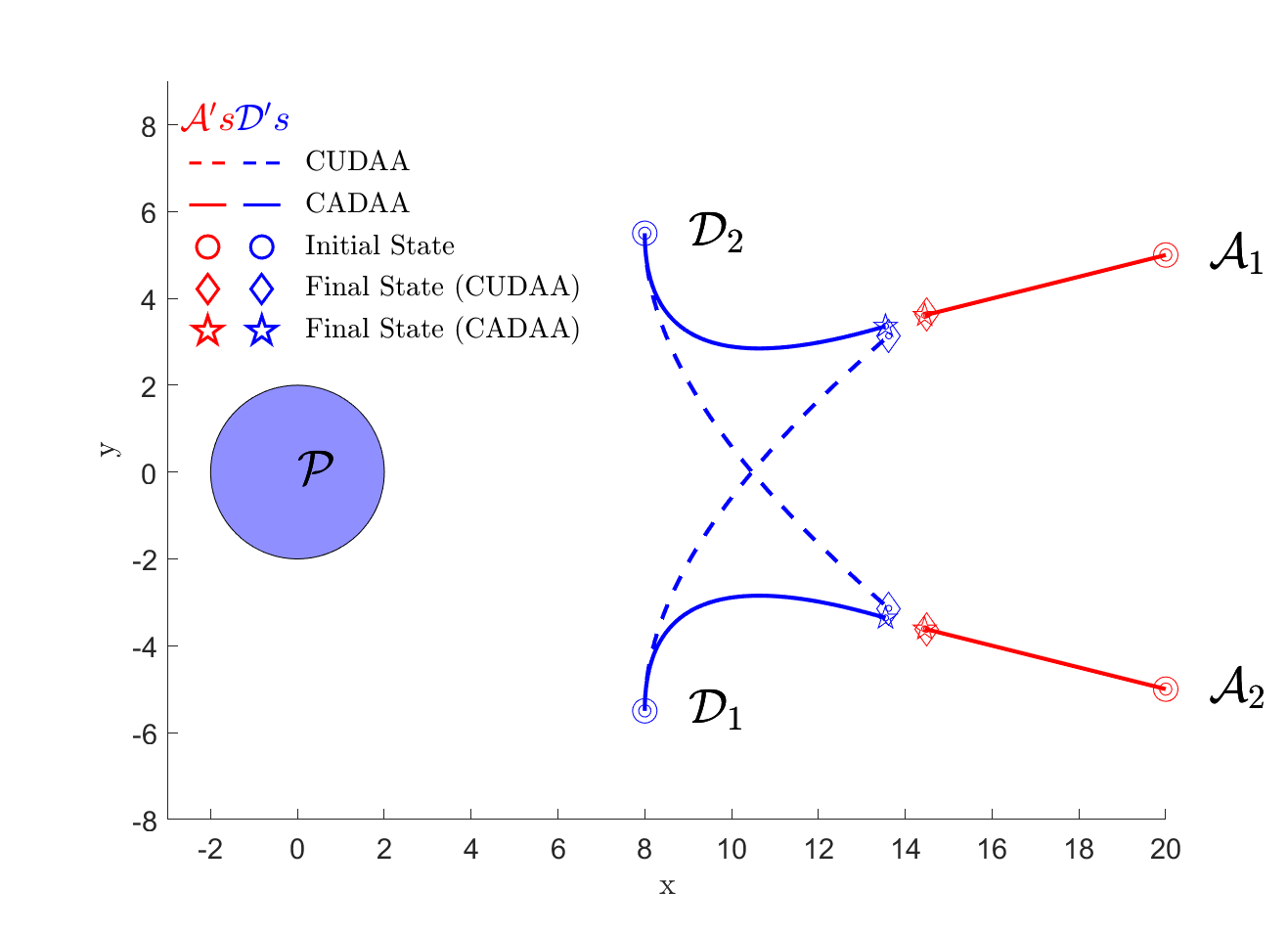}
	\caption{{Scenario 1: CADAA helps avoid inter-defender collision}}
\label{fig:assign_compare_no_collision_def_success}
\end{subfigure}
\begin{subfigure}[ht]{0.95\linewidth}
\includegraphics[width=1\linewidth,trim={1cm 0.35cm 1.2cm 1cm},clip]{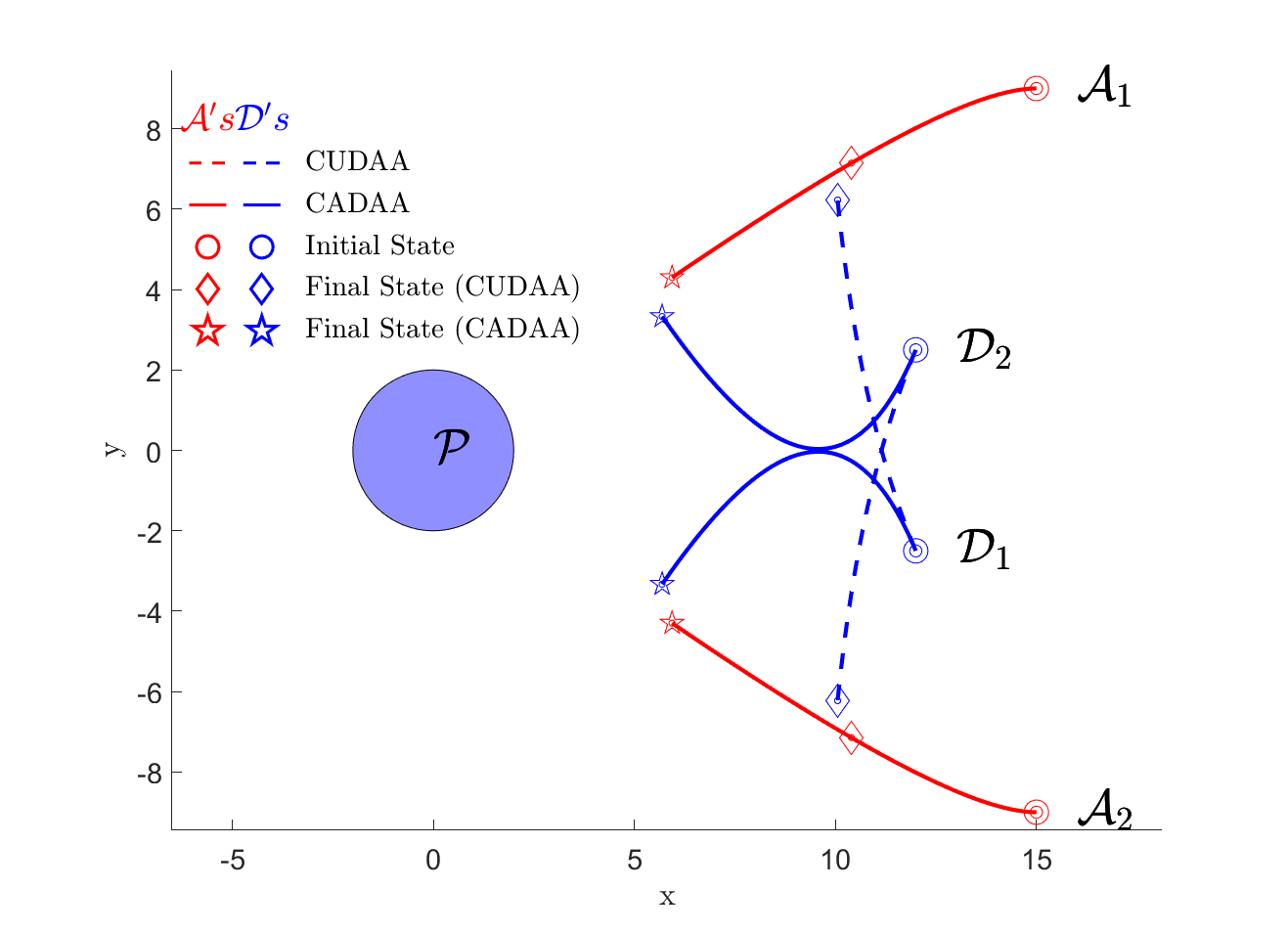}
	\caption{{Scenario 2: CADAA only delays inter-defender collision}}
\label{fig:assign_compare_collision_def_success}
\end{subfigure}
\caption{{Trajectories of the players under CADAA and CUDAA (blue: defenders, red: attackers)}}
\label{fig:assign_compare}
\end{figure}

\begin{figure}
  \centering
  \begin{subfigure}[ht]{0.85\linewidth}
\includegraphics[width=1\linewidth,trim={0.2cm .1cm 0cm 0cm},clip]{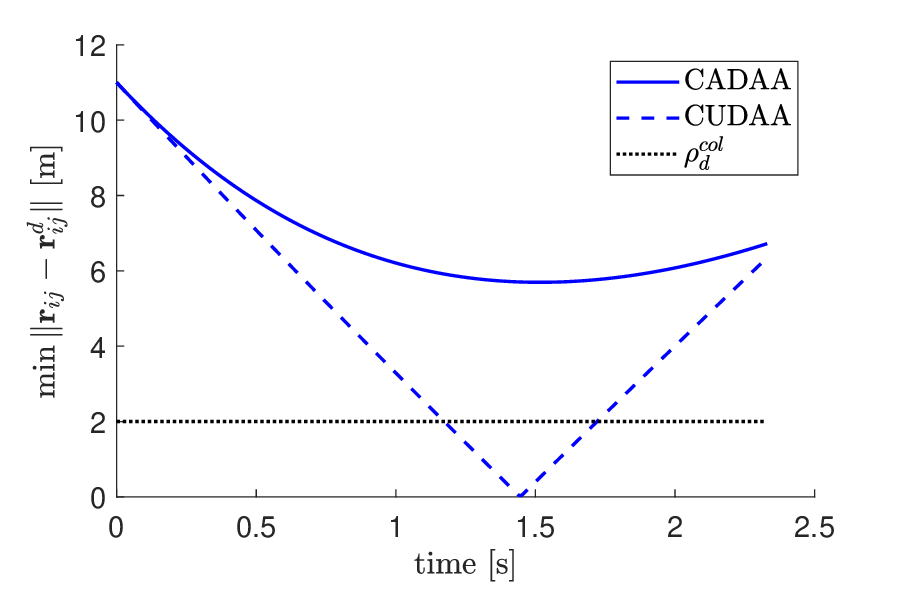}
	\caption{{Minimum distance between the defenders for the trajectories in Scenario 1}}
\label{fig:assign_compare_no_collision_def_success_dist}
\end{subfigure}
\begin{subfigure}[ht]{0.85\linewidth}
\includegraphics[width=1\linewidth,trim={0.2cm .1cm 0cm 0cm},clip]{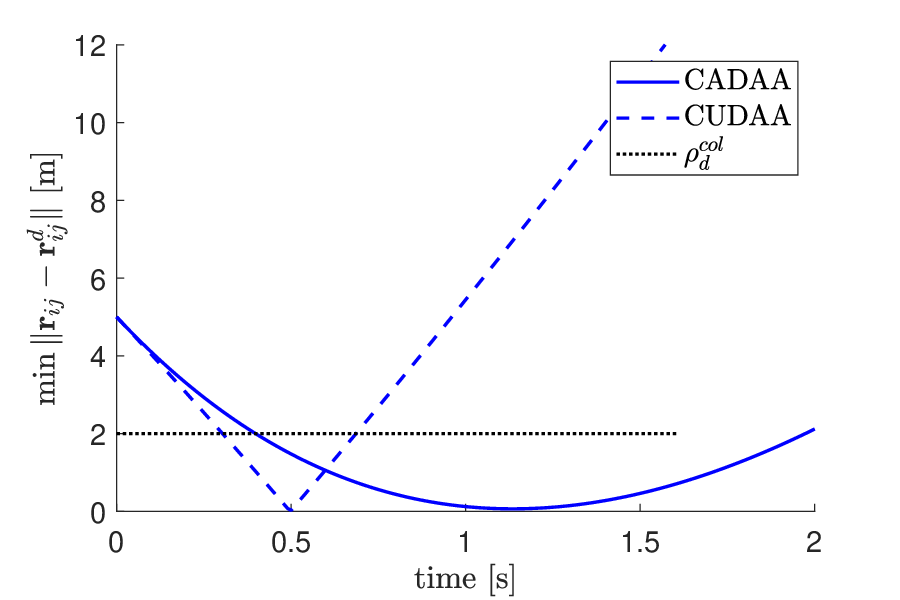}
	\caption{{Minimum distance between the defenders for the trajectories in Scenario 2}}
\label{fig:assign_compare_collision_def_success_dist}
\end{subfigure}
\caption{\vishnu{Minimum distance between defenders under CADAA and CUDAA}}
\label{fig:assign_compare_dist}
\end{figure}

Since it is impossible to avoid collisions in some cases even after accounting for future collisions during the assignment stage, in \vishnu{Fig.~}\ref{fig:collision_avoid_compare}, we show how the ECBF-QCQP based correction helps defenders avoid collisions with each other. In \vishnu{Fig.~}\ref{fig:collision_avoid_compare}(a) {Scenario 3} is shown where despite CADAA the defenders would collide (dashed paths), but after the ECBF-QCQP corrections the defenders avoid collisions (solid paths in \vishnu{Fig.~}\ref{fig:collision_avoid_compare}(a), see also distances in \vishnu{Fig.~}\ref{fig:collision_avoid_compare_dist}(a)) and are also able to intercept the attackers before the attackers reach the protected area. Due to active collision avoidance among the defenders, the interception of the attackers is delayed compared to the case when there is no active collision avoidance used by the defenders. This may result into scenarios where the defenders engage in collision avoidance but fail to intercept the attackers. An example of such scenarios is shown in \ref{fig:collision_avoid_compare}(b), {Scenario 4}, where the defenders avoid collision with each other (see \vishnu{Fig.~}\ref{fig:collision_avoid_compare_dist}(b) for distances) but they are unable to intercept both the attackers in time, attacker $\calA_1$ is able to reach the protected area before the defender $\calD_2$ could intercept it (see \vishnu{Fig.~}\ref{fig:collision_avoid_compare_dist}(b)).
\begin{figure}[ht]
\centering
\begin{subfigure}[ht]{0.98\linewidth}
\includegraphics[width=1\linewidth,trim={1.2cm 2.3cm 0.8cm 1.5cm},clip]{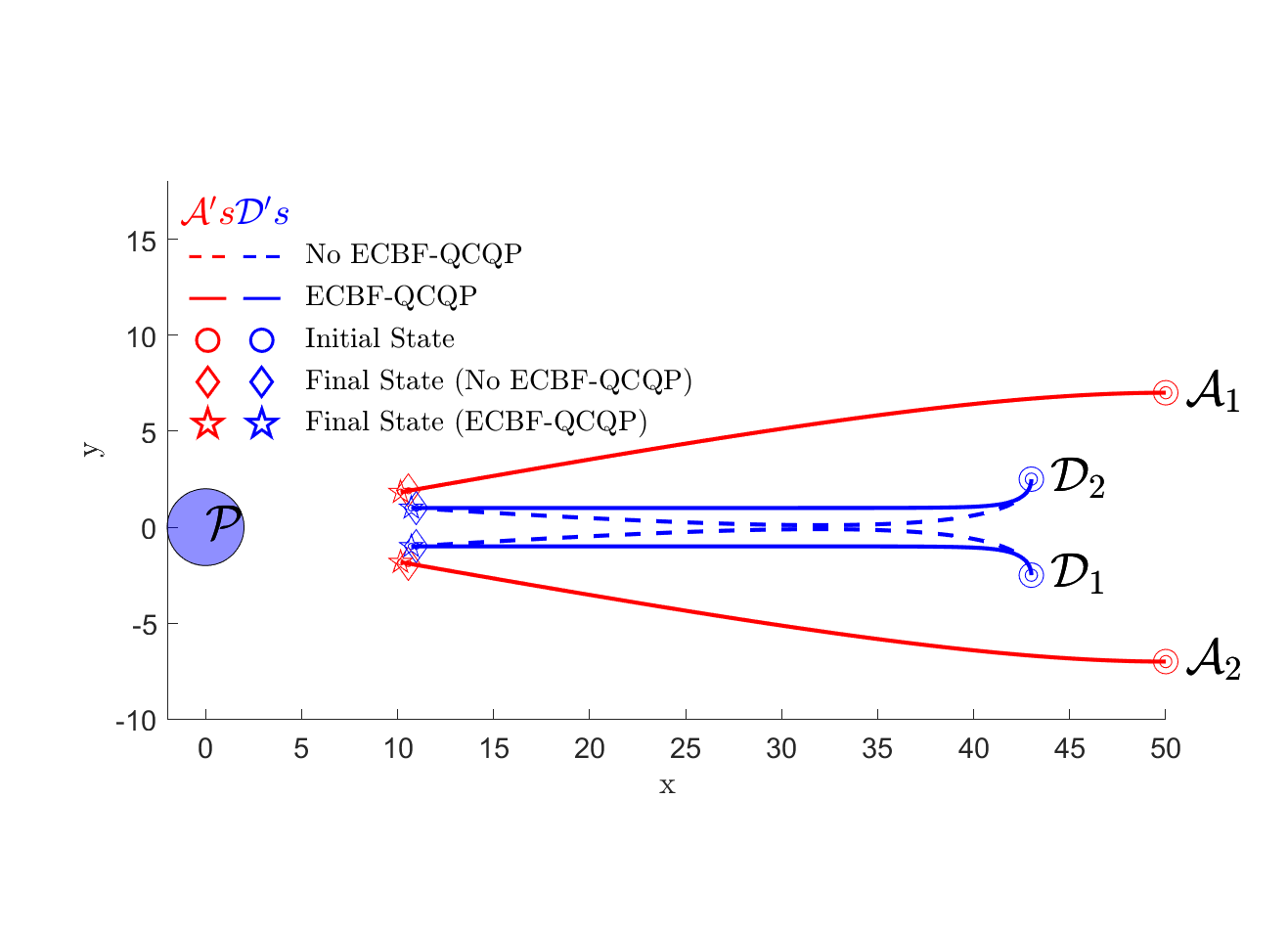}
	\caption{{Scenario 3: Defenders avoid collision with each other and also capture the attackers}}
\label{fig:collision_avoid_compare_no_collision_def_success}
\end{subfigure}
\begin{subfigure}[ht]{0.98\linewidth}
\includegraphics[width=1\linewidth,trim={1.2cm 1.0cm 0cm 1.95cm},clip]{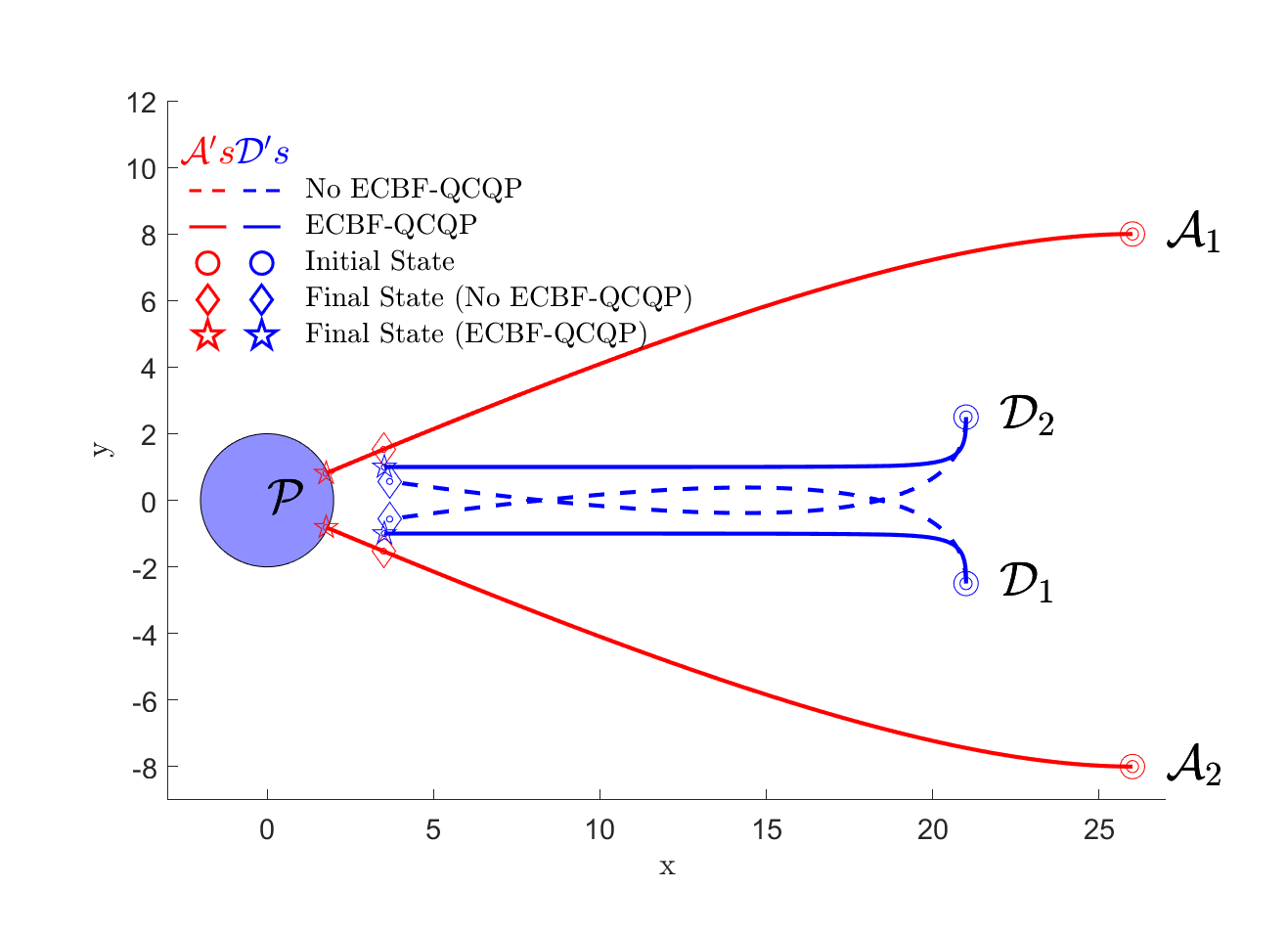}
\caption{{Scenario 4: Defenders avoid collision with each other but fail to capture the attackers}}
\label{fig:collision_avoid_compare_no_collision_def_fail}
\end{subfigure}
\caption{{Trajectories of the players under CADAA with collision avoidance control (ECBF-QCQP), i.e., IDCAIS, and without collision avoidance control (no ECBF-QCQP) (blue: defenders, red: attackers)}}
\label{fig:collision_avoid_compare}
\end{figure}

\begin{figure} 
  \centering
\begin{subfigure}[ht]{0.8\linewidth}
\includegraphics[width=1\linewidth,trim={0.5cm .0cm 0.9cm .2cm},clip]{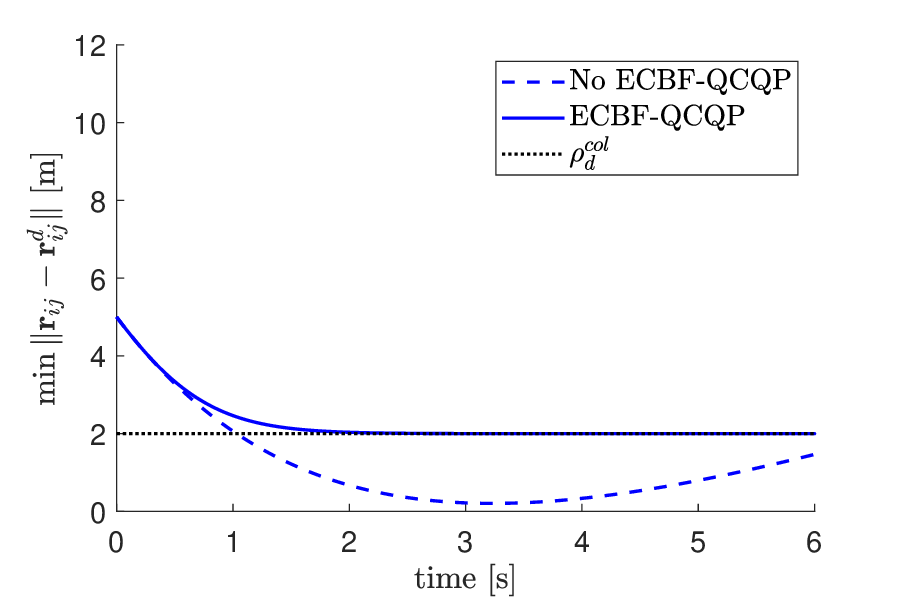}
	\caption{{Minimum distance between the defenders for the trajectories in Scenario 3)}}
\label{fig:collision_avoid_compare_no_collision_def_success_distDD}
\end{subfigure}
\begin{subfigure}[ht]{0.8\linewidth}
\includegraphics[width=1\linewidth,trim={0.5cm .0cm 1cm .0cm},clip]{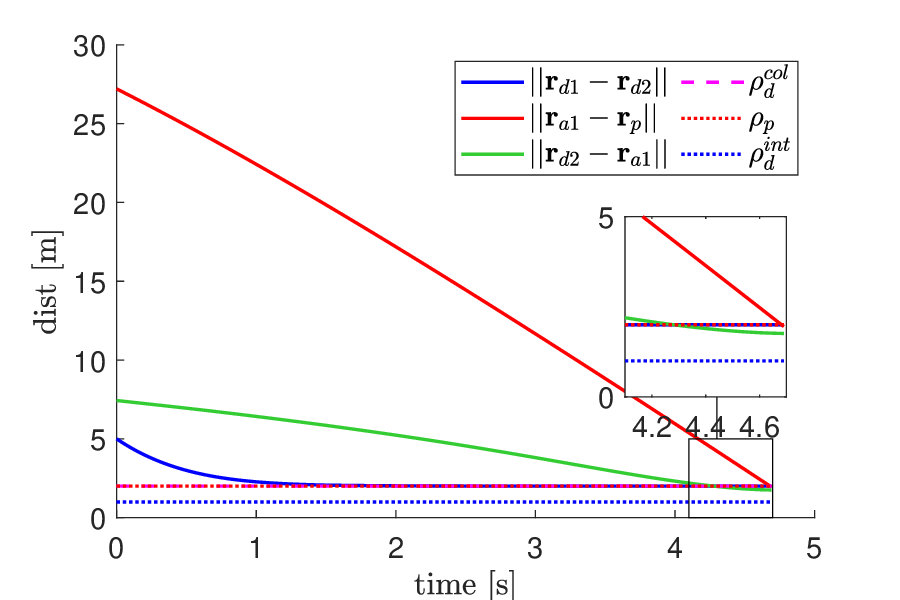}
	\caption{{Minimum distance between the defenders for the trajectories in Scenario 4}}
\label{fig:collision_avoid_compare_no_collision_def_fail_distDD}
\end{subfigure}
\caption{{\small{Minimum distance between the defenders under IDCAIS (CADAA \& ECBF-QCQP) (blue: defenders, red: attackers)}}}
\label{fig:collision_avoid_compare_dist}
\end{figure}

\subsection{{Attackers use time-sub-optimal control}} 
{Sometimes the attackers may not use time-optimal control for their motion to the protected area, for instance in an attempt to escape from the defenders. We ran two simulations (Scenario 5 and 6) in which attackers start playing time-suboptimal action after some time. Due to limited space, the results are provided in the video (link provided later). In Scenario 5, a {time-sub-optimal} behavior of the attackers hampers their chances of hitting the protected area. The attacker $\calA_3$ could have succeeded in reaching the protected area had it played optimally. However, $\calA_3$'s attempt to escape from $\calD_3$ gives more time to $\calD_3$ to capture $\calA_3$, and $\calD_3$ indeed captures $\calA_3$ while all defenders stay safe with respect to each other.} 

{In {Scenario 6}, again due to the active inter-defender collision avoidance, the defender $\calD_3$ fails to intercept $\calA_3$ despite the fact that: 1) $\calA_3$ {uses time-sub-optimal control} in an attempt to escape from the defenders, and 2) had $\calA_3$ played time-optimally, it would have been captured by $\calD_3$. This shows that while having an active inter-defender collision is beneficial for the safety of the defenders, it also compromises the interception performance. }

{A video of the above discussed simulations can be found at \href{https://tinyurl.com/yduuruff}{https://tinyurl.com/yduuruff}.}

\subsection{{Performance of CADAA}}
	\begin{figure}[ht]
   		\centering
		\includegraphics[width=.85\linewidth,trim={0.3cm 0cm 0.3cm 0.2cm},clip]{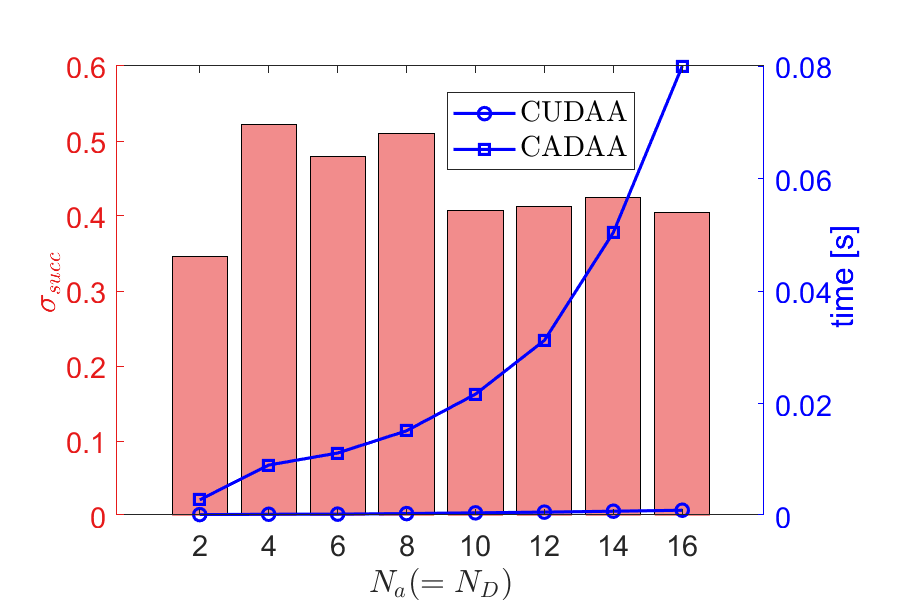}
		\caption{{Computational time of the proposed CADAA (MIQP) vs CUDAA (Hungarian algorithm \cite{kuhn1955hungarian}) and corresponding success rates (for each given number of attackers, 500 random initial configurations are chosen to compute the corresponding $\sigma_{succ}$ value)}}
		\label{fig:performance_compare_CADAA_vs_CUDAA}
	\end{figure}
	{ In this section we provide a numerical evaluation of the computational cost and the effectiveness of the proposed CADAA against the baseline CUDAA (which is usually solved in the literature using Hungarian algorithm \cite{kuhn1955hungarian}). It is numerically intractable to evaluate all possible scenarios even for 2 defenders vs 2 attackers case. So, we solve CADAA and CUDAA problems for random initial configurations of the players for varying number of attackers on a computer with 16 core Intel-i7 processor and 64 GB RAM using MATLAB. In \vishnu{Fig.~}\ref{fig:performance_compare_CADAA_vs_CUDAA}, we show the computation times to solve CADAA and CUDAA for varying number of attackers (also equals number of defenders). Each data point in \vishnu{Fig.~}\ref{fig:performance_compare_CADAA_vs_CUDAA} is obtained by solving the CADAA and CUDAA for 500 random initial configurations of the players for a given number of attackers. In \vishnu{Fig.~}\ref{fig:performance_compare_CADAA_vs_CUDAA}, for a given number of attackers, we also show using bar chart the success rate of the CADAA against the CUDAA defined as $$\sigma_{succ} = \frac{\text{ \# cases CADAA is able to avoid collisions}}{\text{\# cases CUDAA results into defender collisions}}.$$ }
	
	{
	As shown in \vishnu{Fig.~}\ref{fig:performance_compare_CADAA_vs_CUDAA}, the computation cost of the CADAA because of the underlying MIQP grows exponentially as the number of attackers grow while that of the CUDAA has polynomial dependence on the number of attackers due to the underlying linear assignment problem. While it is evident that the computation cost of the proposed CADAA is significantly higher than the baseline CUDAA approach that is prevalent in the literature, on average CADAA helps avoid inter defender collisions in around $40\%$ cases just by taking into account the future possible collisions among the defenders. In summary, the proposed CADAA provides a inter-defender collision free defender-to-attacker assignment whenever feasible at the cost of additional computation time.}
	
	\subsection{Additional Simulations} \label{sec:additional_simulations}
	{We also demonstrate the proposed algorithm for case studies involving up to 16 defenders against 16 attackers. In the interest of space, simulations of these cases studies are provided in the video available at \href{https://tinyurl.com/yduuruff}{https://tinyurl.com/yduuruff}. The considered scenarios are discussed as follows.
	\bi
	\item Scenario 7: There are 8 attackers approaching the protected area and 8 defenders located around the protected area. The CUDAA results into an assignment of defenders to attackers that results into defenders $\calD_3$ and $\calD_7$ colliding with each other (see the video). On the other hand, IDCAIS results into defenders successfully capturing all the attackers before the attackers reach the protected while ensuring no collision occurred between the fellow defenders.
	\item Scenario 8: There are 16 attackers approaching the protected area and 16 defenders located around the protected area. The CUDAA results into an assignment of defenders to attackers that results into defenders $\calD_1$ and $\calD_6$ colliding with each other (see the video). On the other hand, IDCAIS results into defenders successfully capturing all the attackers before the attackers reach the protected while ensuring no collision occurred between the fellow defenders.
	\ei 
	}


\section{Conclusions and Future Work} \label{sec:conclusions}
In this paper, we proposed inter-defender collision-aware interception strategy for a team of defenders moving under damped double integrator models to intercept multiple attackers. We used time-optimal, collision-aware, defender-to-attacker assignment (CADAA) to optimally assign defenders to attackers, and exponential-CBF-based QCQP for active inter-defender collision avoidance. The effectiveness and limitations of the proposed strategy are demonstrated formally as well as through simulations. The proposed collision-aware assignment (CADAA) requires solving a Mixed-Integer Quadratic Program (MIQP), which is computationally challenging as the number of the involved players increases. 
\xinyi{One limitation is the difficulty in identifying initial configurations that guarantee collision-free defender assignments. As future work, we plan to develop scalable analytic or data-driven criteria to address this challenge.}
Another future work is that we plan to investigate computationally-efficient heuristics in orde to solve the MIQP presented in this paper, and therefore reduce the computation time of CADAA.
{We would also like to investigate how the assignment would be affected in the case of measurement and environmental uncertainty by performing a formal robustness analysis in future work.
}	
	
\appendix
\subsection{Velocity bounds}\label{append:velocity_bounds}
\begin{lemma}
Consider the system in Eq. \eqref{eq:dampedDIDyn},
where \(C_D > 0\) is a drag coefficient and the input \(\mathbf{u}_{\imath}(t) \in \mathbb{R}^2\) satisfies
\[
\|\mathbf{u}_{\imath}(t)\| \leq \bar{u} \quad \text{for all } t \geq 0.
\]
If the initial condition satisfies \(\|\mathbf{v}_{\imath 0}\| \leq \frac{\bar{u}}{C_D}\), then the solution \(\mathbf{v}_{\imath}(t)\) satisfies
\[
\|\mathbf{v}_{\imath}(t)\| \leq \frac{\bar{u}}{C_D} \quad \text{for all } t \geq 0.
\]
\end{lemma}
\begin{proof}
From Eq.~\eqref{eq:dampedDIDyn}, the dynamics of the velocity can be extracted as:
\[
\dot{\mathbf{v}}_{\imath}(t) = -C_D \mathbf{v}_{\imath}(t) + \mathbf{u}_{\imath}(t), \hspace{1mm} \mathbf{v}_{\imath} (0)= \mathbf{v}_{\imath0} \in \bR^2
\]
The solution of the above differential equation is given by
\[
\mathbf{v}_{\imath}(t) = e^{-C_D t} \mathbf{v}_{\imath 0} + \int_0^t e^{-C_D (t - \tau)} \mathbf{u}_{\imath}(\tau) \, d\tau.
\]
Taking the norm on both sides and applying the triangle inequality:
\[
\|\mathbf{v}_{\imath}(t)\| \leq e^{-C_D t} \|\mathbf{v}_{\imath 0}\| + \int_0^t e^{-C_D (t - \tau)} \|\mathbf{u}_{\imath}(\tau)\| \, d\tau.
\]
Using the bound \(\|\mathbf{u}_{\imath}(\tau)\| \leq \bar{u}\), we have:
\[
\|\mathbf{v}_{\imath}(t)\| \leq e^{-C_D t} \|\mathbf{v}_{\imath 0}\| + \bar{u} \int_0^t e^{-C_D (t - \tau)} \, d\tau.
\]
Evaluating the integral:
\[
\int_0^t e^{-C_D (t - \tau)} \, d\tau = \int_0^t e^{-C_D s} \, ds = \frac{1 - e^{-C_D t}}{C_D}.
\]
Substituting back:
\[
\|\mathbf{v}_{\imath}(t)\| \leq e^{-C_D t} \|\mathbf{v}_{\imath 0}\| + \frac{\bar{u}}{C_D} (1 - e^{-C_D t}).
\]
Now, using \(\|\mathbf{v}_{\imath 0}\| \leq \frac{\bar{u}}{C_D}\):
\[
\|\mathbf{v}_{\imath}(t)\| \leq \frac{\bar{u}}{C_D} \left( e^{-C_D t} + 1 - e^{-C_D t} \right) = \frac{\bar{u}}{C_D}.
\]
\end{proof}

\subsection{Existence and uniqueness of solution to Eq.~\eqref{eq:optimal_control_soln}}\label{append:existence_and_uniqueness}
\begin{lemma} Consider an agent with linear drag dynamics \eqref{eq:dampedDIDyn}, starting from position $\mathbf{r}_{\imath 0} = [x_{\imath 0}, y_{\imath 0}]^\top$ and velocity $\mathbf{v}_{\imath 0} = [v_{x_{\imath 0}}, v_{y_{\imath 0}}]^\top$, subject to a control of constant magnitude $\bar{u}$ and direction $\theta_{\imath}^*$ with the terminal condition, at some $t_f>0$, on the boundary of the circle of radius $\rho_{\dag}$ centered at $\mathbf{r}_{\dag} = [x_{\dag}, y_{\dag}]^\top$ defined as $\mathcal{B}_{\rho_{\dag}}(\mathbf{r}_{\dag}) = \{\mathbf{r} \in \bR^2| \norm{\mathbf{r}-\mathbf{r}_{\dag}}<\rho_{\dag} \}$. Then, the system of equations \eqref{eq:optimal_control_soln} has a unique solution \( (t_f, \theta_{\imath}^*) \in \bR_+ \times [0, 2\pi) \), provided \( \mathbf{r}_{\imath 0} \notin \mathcal{B}_{\rho_{\dag}}(\mathbf{r}_{\dag})\) (i.e., the initial state is not inside the terminal circle).
	\end{lemma}
\begin{proof} The system can be written in vector form as:
	\[
	\mathbf{r}_{\dag} - \rho_{\dag} \mathbf{e}(\theta_{\imath}^*) = \mathbf{r}_{\imath 0} + E_1(t_f) \mathbf{v}_{\imath 0} + E_2(t_f) \bar{u} \mathbf{e}(\theta_{\imath}^*),
	\]
	where $\mathbf{e}(\theta) := [\cos\theta, \sin\theta]^\top$. Rearranging:
	\[
	\mathbf{r}_{\dag} - \mathbf{r}_{\imath 0} - E_1(t_f) \mathbf{v}_{\imath 0} = \left( E_2(t_f) + \rho_{\dag} \right) \mathbf{e}(\theta_{\imath}^*).
	\]
	Taking norms on both sides:
	\[
	\left\| \mathbf{r}_{\dag} - \mathbf{r}_{\imath 0} - E_1(t_f) \mathbf{v}_{\imath 0} \right\| = E_2(t_f) + \rho_{\dag}.
	\]
	Define the scalar function:
	\[
	D(t_f) := \left\| \mathbf{r}_{\dag} - \mathbf{r}_{\imath 0} - E_1(t_f) \mathbf{v}_{\imath 0} \right\| - \left( E_2(t_f) + \rho_{\dag} \right).
	\]
	As $t_f \to 0^+$:
	\[
	D(t_f) \to \| \mathbf{r}_{\dag} - \mathbf{r}_{\imath 0} \| - \rho_{\dag} > 0,
	\]
	since $\mathbf{r}_{\imath 0}$ lies outside the terminal circle.
	
	As $t_f \to \infty$, we have $E_1(t_f) \to 1/C_D$, $E_2(t_f) \to \infty$, so:
	\[
	D(t_f) \to \text{const} - \infty \to -\infty.
	\]
	By the Intermediate Value Theorem, there exists $t_f^* > 0$ such that $D(t_f^*) = 0$. This implies the vector on the left has the same norm as the right, so we can recover the heading:
	\[
	\theta_{\imath}^* = \arg\left( \mathbf{r}_{\dag} - \mathbf{r}_{\imath 0} - E_1(t_f^*) \mathbf{v}_{\imath 0} \right).
	\]
	Hence, a solution $(t_f, \theta_{\imath}^*)$ exists.\\

	\noindent \textbf{Uniqueness:} We now prove \( D(t_f) \) has a unique root by showing it is unimodal and strictly decreasing after a maximum.
		Let
		\[
		\mathbf{R}(t_f) := \mathbf{r}_{\dag} - \mathbf{r}_{\imath 0} - E_1(t_f) \mathbf{v}_{\imath 0},
		\]
		so
		\[
		D(t_f) = \|\mathbf{R}(t_f)\| - E_2(t_f) - \rho_{\dag}.
		\]
		Taking the derivative:
		\[
		\frac{d}{dt_f} D(t_f) = -\frac{d E_1}{dt_f} \cdot \frac{\mathbf{R}(t_f)^\top \mathbf{v}_{\imath 0}}{\|\mathbf{R}(t_f)\|} - \frac{d E_2}{dt_f},
		\]
		with \( \frac{d E_1}{dt_f} = e^{-C_D t_f} \), \( \frac{d E_2}{dt_f} = \frac{1 - e^{-C_D t_f}}{C_D} \). Therefore:
		\[
		\frac{d}{dt_f} D(t_f) = - e^{-C_D t_f} \cdot \frac{\mathbf{R}(t_f)^\top \mathbf{v}_{\imath 0}}{\|\mathbf{R}(t_f)\|} - \frac{1 - e^{-C_D t_f}}{C_D}.
		\]
		
		Let us bound the first term:
		\[
		\left| \frac{\mathbf{R}(t_f)^\top \mathbf{v}_{\imath 0}}{\|\mathbf{R}(t_f)\|} \right| \leq \|\mathbf{v}_{\imath 0}\|,
		\quad \text{(by Cauchy-Schwarz)}.
		\]
		Thus:
		\[
		\frac{d}{dt_f} D(t_f) \leq e^{-C_D t_f} \|\mathbf{v}_{\imath 0}\| - \frac{1 - e^{-C_D t_f}}{C_D}.
		\]
		This function is positive for small \( t_f \), but strictly negative for large \( t_f \), and has a unique maximum. Therefore, \( D(t_f) \) is unimodal and crosses zero at most once. Hence, the solution $(t_f, \theta_{\imath}^{*})$ to the system of equations in \eqref{eq:optimal_control_soln} is unique.

	\end{proof}

\subsection{Finite direction changes of distance evolution}\label{append:finite_direction_changes_of_distance}
\begin{lemma}
Consider two agents $\imath$ and $\imath'$ governed by the linear drag dynamics:
\begin{align}
    \dot{\mathbf{r}}_{\imath}(t) &= \mathbf{v}_{\imath}(t), \quad
    \dot{\mathbf{v}}_{\imath}(t) = \mathbf{u}_{\imath} - C_D \mathbf{v}_{\imath}(t), \\
    \dot{\mathbf{r}}_{\imath'}(t) &= \mathbf{v}_{\imath'}(t), \quad
    \dot{\mathbf{v}}_{\imath'}(t) = \mathbf{u}_{\imath'} - C_D \mathbf{v}_{\imath'}(t),
\end{align}
with $C_D > 0$ and constant control inputs $\mathbf{u}_{\imath}, \mathbf{u}_{\imath'} \in \mathbb{R}^2$. Let $\mathbf{r}_{\Delta}(t) = \mathbf{r}_{\imath}(t) - \mathbf{r}_{\imath'}(t)$ denote the relative position between the two agents, and define the inter-agent distance $d(t) = \|\mathbf{r}_{\Delta}(t)\|$.

Then, the number of turning points (i.e., zeros of $\dot{d}(t)$) in $d(t)$ is at most four on the interval $[0, \infty)$.
\end{lemma}

\begin{proof}
For each agent, under constant control, the velocity is given as:
\begin{equation}
    \mathbf{v}_{\imath}(t) = \mathbf{v}_{\infty,\imath} + (\mathbf{v}_{\imath0} - \mathbf{v}_{\infty,\imath}) e^{-C_D t}, \quad \text{where} \ \mathbf{v}_{\infty,\imath} = \frac{\mathbf{u}_{\imath}}{C_D}.
\end{equation}
The position is obtained by integrating the velocity:
\begin{equation}
    \mathbf{r}_{\imath}(t) = \mathbf{r}_{\imath0} + \mathbf{v}_{\infty,\imath} t + \frac{1}{C_D} (\mathbf{v}_{\imath0} - \mathbf{v}_{\infty,\imath}) (1 - e^{-C_D t}).
\end{equation}
Define the relative position $\mathbf{r}_{\Delta}(t) = \mathbf{r}_{\imath}(t) - \mathbf{r}_{\imath'}(t)$. Then,
\begin{align}
    \mathbf{r}_{\Delta}(t) &= \mathbf{a} + \mathbf{b} t + \mathbf{c}(1 - e^{-C_D t}), \\
    \dot{\mathbf{r}}_{\Delta}(t) &= \mathbf{b} + \mathbf{c} C_D e^{-C_D t},
\end{align}
where
\[
\mathbf{a} = \mathbf{r}_{\imath0} - \mathbf{r}_{\imath'0}, \quad
\mathbf{b} = \mathbf{v}_{\infty,\imath} - \mathbf{v}_{\infty,\imath'}, \quad
\]
\[
\mathbf{c} = \frac{1}{C_D} \left[(\mathbf{v}_{\imath0} - \mathbf{v}_{\infty,\imath}) - (\mathbf{v}_{\imath'0} - \mathbf{v}_{\infty,\imath'})\right].
\]

The time derivative of the distance function is:
\begin{equation}
    \dot{d}(t) = \frac{\mathbf{r}_{\Delta}(t)^\top \dot{\mathbf{r}}_{\Delta}(t)}{\|\mathbf{r}_{\Delta}(t)\|}.
\end{equation}
Let $f(t) = \mathbf{r}_{\Delta}(t)^\top \dot{\mathbf{r}}_{\Delta}(t)$, which has the same sign as $\dot{d}(t)$ (assuming $\|\mathbf{r}_\Delta(t)\| \ne 0$). Then,
\begin{equation}
    f(t) = (\mathbf{a} + \mathbf{b} t + \mathbf{c} (1 - e^{-C_D t}))^\top (\mathbf{b} + \mathbf{c} C_D e^{-C_D t}).
\end{equation}

Expanding this inner product gives a function $f(t)$ composed of the following basis functions:
\[
1, \quad t, \quad e^{-C_D t}, \quad t e^{-C_D t}, \quad e^{-2C_D t}.
\]
These five functions are linearly independent on $[0, \infty)$. Therefore, $f(t)$ belongs to the finite-dimensional function space spanned by these five basis functions.

Next, we consider their Wronskian matrix:
\[
W_4(t) =
\begin{bmatrix}
1 & t & E & tE & E^2 \\
0 & 1 & -C_D E & E(1 - C_D t) & -2C_D E^2 \\
0 & 0 & C_D^2 E & E(C_D^2 t - 2 C_D) & 4 C_D^2 E^2 \\
0 & 0 & -C_D^3 E & E(3 C_D^2 - C_D^3 t) & -8 C_D^3 E^2 \\
0 & 0 & C_D^4 E & E(C_D^4 t - 4 C_D^3) & 16 C_D^4 E^2
\end{bmatrix}.
\]
where $E = e^{-C_D t}$.
Evaluating the determinant of $W_4(t)$ yields:
\[
\det (W_4(t)) = 4 C_D^8 e^{-4 C_D t},
\]
which is strictly positive for all \( t \geq 0 \) and \( C_D > 0 \). Similarly, for other Wronskian matrices $W_k$ (formed by removing all columns $j$ and rows $j$, for $j>k+1$, from $W_4(t)$) for all $k \in \{3,2,1,0\}$, the determinants are:
\[
W_3 = C_D^4 e^{-2C_Dt}, \; W_2 = C_D^2 e^{-C_Dt},\; W_1=1, \;W_0=1
\]
Therefore, the Wronskian $W_k$ is non-zero for all $k\in \{4,3,2,1,0\}$, and the set of functions forms a Chebyshev system on \( [0, \infty) \).

Consequently, from the theory of Chebyshev systems based on Theorem 70 in \cite{meinardus2012approximation} any nontrivial linear combination of these five functions, such as \( f(t) \), can have at most four distinct real zeros on \( [0, \infty) \). Each such zero corresponds to a turning point in the inter-agent distance \( d(t) \), $\mathbf{r}_\Delta(t) \ne 0$.



Therefore, the number of turning points of $d(t)$ is at most 4.
\end{proof}

    \bibliographystyle{IEEEtran}
    \bibliography{refs}

\vspace{-2cm}
\end{document}